\documentclass[%
 reprint,
nofootinbib,
 amsmath,amssymb,
 aps,
]{revtex4-2}

\usepackage{graphicx}
\usepackage{dcolumn}
\usepackage{bm}

\usepackage[utf8]{inputenc}
\usepackage[english]{babel}

\usepackage{dsfont, xcolor, amsmath, amscd, amssymb, fancyhdr, amsfonts, amsbsy, amsopn, amsthm, url, color}

\usepackage{array}   
\newcolumntype{C}{>{$}c<{$}} 

\newtheorem{claim}{Claim}
\newtheorem{computation}{Computation}
\newtheorem{theorem}{Theorem}
\newtheorem{lemma}{Lemma}
\newtheorem{defn}{Definition}
\newtheorem{proposition}{Proposition}
\newtheorem{newrule}{Rule}
\newtheorem{example}{Example}
\newtheorem{corollary}{Corollary}
\newtheorem{condition}{Condition}

\usepackage[font=footnotesize]{caption}
\usepackage[colorlinks=true,linkcolor=cyan,citecolor=blue]{hyperref}

\usepackage{todonotes}

\newcommand{\etal}{\text{et al.\@ }}

\DeclareMathOperator{\Sgn}{Sgn}
\newcommand{\sgn}{\Sgn}

\newcommand{\XX}{{\bm{X}}} 
\newcommand{\YY}{{\bm{Y}}} 
\newcommand{\ZZ}{{\bm{Z}}}
\newcommand{\xx}{{\bm{x}}} 
\newcommand{\yy}{{\bm{y}}} 

\newcommand{\PM}{\text{PM}}
\newcommand{\MIN}{\text{min}}

\newcommand{\Ibroja}{I_{\cap}^{\text{BROJA}}}
\newcommand{\Ired}{I_{\cap}}
\newcommand{\Imin}{I_{\cap}^{\MIN}}

\newcommand{\Iccs}{I_{\cap}^{CCS}}
\newcommand{\Ipm}{I_{\cap}^{\PM}}

\newcommand{\Ipmp}{I_{\cap}^{\PM,+}}
\newcommand{\Ipmm}{I_{\cap}^{\PM,-}}
\newcommand{\Ipmpm}{I_{\cap}^{\PM,\pm}}
\newcommand{\isx}{i_{\cap}^{\text{sx}}}
\newcommand{\Isx}{I_{\cap}^{\text{sx}}}

\newcommand{\Smin}{S^{\MIN}}

\newcommand{\Spm}{S^{\PM}}

\newcommand{\Rmin}{R^{\MIN}}

\newcommand{\Rpm}{R^{\PM}}

\newcommand{\Umin}[1]{U_{#1}^{\MIN}}

\newcommand{\Upm}[1]{U_{#1}^{\PM}}
\newcommand{\Upmp}[1]{U_{#1}^{\PM,+}}
\newcommand{\Upmm}[1]{U_{#1}^{\PM,-}}

\newcommand{\oeq}{\stackrel{\omega}{=}}

\newcommand{\dmin}{\mathfrak{d}_{\MIN}}
\newcommand{\dpm}{\mathfrak{d}_{\PM}}

\newcommand{\beginsupplement}{%
        \setcounter{table}{0}
        \renewcommand{\thetable}{S\arabic{table}}%
        \setcounter{figure}{0}
        \renewcommand{\thefigure}{S\arabic{figure}}%
     }

\begin{document}

\preprint{APS/123-QED}

\title{Signed and Unsigned Partial Information\\ Decompositions of Continuous Network Interactions}

\author{Jesse Milzman}
 \altaffiliation[]{Army Research Laboratory}
 \email{jesse.m.milzman.civ@army.mil}
\author{Vince Lyzinski}%
\altaffiliation[]{Department of Mathematics, University of Maryland}
 \email{vlyzinsk@umd.edu}




\date{\today}

\begin{abstract}

We investigate the partial information decomposition (PID) framework as a tool for edge nomination.
We consider both the $\Imin$ and $\Ipm$ PIDs, from \cite{Williams2010} and \cite{Finn2018main} respectively, and we both numerically and analytically investigate the utility of these frameworks for discovering significant edge interactions.
In the course of our work, we extend both the $\Imin$ and $\Ipm$ PIDs to a general class of continuous trivariate systems.
Moreover, we examine how each PID apportions information into redundant, synergistic, and unique information atoms within the source-bivariate PID framework.
Both our simulation experiments and analytic inquiry indicate that the atoms of the $\Ipm$ PID have a non-specific sensitivity to high predictor-target mutual information, regardless of whether or not the predictors are truly interacting.
By contrast, the $\Imin$ PID is quite specific, although simulations suggest that it lacks sensitivity.


\end{abstract}

\maketitle


\section{Introduction}
\label{section:introduction}

Information-theoretic methods are non-parametric and well-suited to capturing complex, non-linear dependency between variables.
Although less common than other methods, information theoretic tools have been widely applied in gene network analysis, and mutual information (MI) has commonly been employed in place of Pearson correlation to create association networks (\cite{Butte2000RNAchemo,Butte2003chapter,Chanda2020information}).
However, much like Pearson correlation and other naïve pairwise methods, MI cannot distinguish between direct and indirect associations.

In gene regulatory networks (GRNs) and protein-protein interaction (PPI) networks, especially,
we expect a highly interdependent structure, in which indirect associations are not immediately distinguishable from direct via any pairwise comparison.
As such, pairwise methods have an unacceptably high false discovery rate, and any networks inferred with them will include too many spurious associations for the biologically meaningful ones to stand out \cite{soranzo2007comparing}.
Moreover, complex biological systems are not dyadic, but usually contain
higher order interactions.
To overcome these challenges, the next generation of information methods for biological network inference moved beyond MI in two keys ways.
First, researchers have attempted to remove indirect associations by accounting for redundancy \cite{Timme2014experimentalist}, e.g. \cite{Margolin2006aracne,Meyer2007MRNET,Villaverde2014mider,soranzo2007comparing,Liang2008geneCMI, Zhao2008CMI, Zhang2015conditional}.
Second, attempts were made to  specifically identify interacting genes by quantifying the synergistic information of the pair that cannot be acquired from either gene alone \cite{Anastassiou2007synergy, Watkinson2008identification, Watkinson2009SACLR}.

Developments within information theory have emerged in the past two decades that address both concerns.
First the neuroscience (\cite{Gat1999synergy,Brenner2000synergy}) and later the more general network communities rediscovered the formalism of interaction information, whose first known use was by McGill in \cite{McGill1954multivariate} (see also \cite{Ting1962amount,Bell2003co}).
This quantity, initially termed synergy in the networks literature \cite{Gat1999synergy,Brenner2000synergy,Anastassiou2007synergy}, is a signed extension of MI, distinguishing synergistic and redundant information as opposites.
A positive interaction information indicates synergy, while a negative value indicates redundancy.
However, this quantity assumes that redundancy and synergy are opposing possibilities, and provides no framework for the possibility of the simultaneous presence of redundant and synergistic information between predictor variables \cite{Timme2014experimentalist}.
This would be a major concern in any highly interdependent network, such as a GRN or the human brain, where we might expect a large amount of synergy and redundancy to each occur in distinct contexts for the same pair of nodes.

This limitation motivated the development of the Partial Information Decomposition (PID) framework in the seminal paper by Williams and Beer \cite{Williams2010}.
Disentangling the synergy and redundancy of interaction information into distinct atoms \cite{Williams2010}, PID is also able to highlight the presence of unique information atoms, which would normally be cancelled out in the computation of interaction information.
For the purposes of network inference, the PID framework offers the possibility of controlling the redundancy in interdependent biological network data without neglecting the useful information remaining, be it unique to a single predictor or synergistic among multiple.
The PID framework was met with much enthusiasm within a community of researchers meeting at the intersection of information theory, networks, complex systems, and neuroscience.
However, there is still dispute and no clear consensus over the correct definition of `redundancy', from which all other PID atoms follow.
Multiple alternatives to the original redundancy measure $\Imin$ of \cite{Williams2010} have been proposed (e.g. \cite{Harder2013bivariate, Bertschinger2014,Ince2017,Finn2018main}, and recently \cite{Makkeh2021differentiable,Schick2021partial}, see also \cite{EntropySIOverview2018} for others).
It seems likely that there may not be any definition of redundancy that will be universally appropriate to all or even most applications.

In this work, we extend the $\Imin$ and $\Ipm$ PIDs to nonlinear interactions of continuous variables, and explore their capacity for the nomination of such interactions in synergy networks inspired by network biology.
Crucially, we find that for bivariate interactions, these two PIDs differ in the apportionment of the information discrepancy between the predictors.
For a noise-free system, we make this precise, and provide expressions for the bivariate $\Imin$ and $\Ipm$ PIDs in terms of ratios of the partial derivatives of the interaction kernel.
For our network simulations, this differential handling of information discrepancy results in the the $\Ipm$ PID suffering from the same non-specificity as mutual information, as it is unable to distinguish variable interactions from univariate contributions.

\section{Previous Work} 
\label{section:previousWork}

\subsection{Information Theory and Network Inference}
\label{subsection:PreviousWork.InfoNetworks}

Multiple information theoretic approaches to network inference have been developed over the past two decades. 
A recent overview of information-theoretic methods in computational biology can be found in \cite{Chanda2020information}.
This work investigates the network-inference potential of the Partial Information Decomposition (PID) framework, a branch within the broader information and complexity literature.
Our review of redundancy and synergy, and their connection to PID, owes much to that in \cite{Timme2014experimentalist}.

Mutual information (MI) has been commonly used to build associational networks, as it offers a non-linear, non-parametric alternative to correlation and other traditional statistical options.
One of the early applications of MI to gene network inference was in the development of relevance networks from gene expression data by Butte and Kohane in \cite{Butte1999mutual}.
Originally, they used pairwise Pearson correlation to infer edges \cite{Butte1999unsupervised}, before employing MI in \cite{Butte1999mutual}.
In choosing MI, they posited, among other advantages, that MI would perform better than Pearson correlation at detecting irregularly distributed variables and complex (non-linear) interactions \cite{Butte2003chapter}.
However, pairwise MI suffers from many of the same issues for network inference as pairwise correlation.
Neither can distinguish direct from indirect interactions, which becomes especially problematic in the context of dense gene networks \cite{Margolin2006aracne,soranzo2007comparing,Chanda2020information}.
In particular, when considering cancer-specific gene networks, there is good reason to suspect more cross-talk, i.e. more active network edges, and thus more indirect interaction than in normal biomolecular networks\footnote{See, for instance, the network entropy literature \cite{West2012differential,Teschendorff2014signalling,Teschendorff2015scalefree,teschendorff2017singleCell}, which suggests more active edges in the PPI network for cancerous samples, as compared to healthy controls. }.
When networks are inferred via isolated pairwise comparisons, Type II errors become a structural concern, rather than a statistical problem to be overcome with more data.

Thus, the next generation of information-theoretic inference methods aimed at controlling FDR due to indirect associations.
A head-to-head comparison of a few of these, compared to traditional MI relevance networks, can be found in \cite{EmmertStreib2012}.
Some of the more well-known methods from this next generation include ARACNE \cite{Basso2005reverse,Margolin2006aracne}, MRNET \cite{Meyer2007MRNET}, and MIDER \cite{Villaverde2014mider}.
Similar to the entropy reduction criterion in MIDER, other methods employed conditional mutual information (CMI) to control indirect associations \cite{soranzo2007comparing, Liang2008geneCMI, Zhao2008CMI, Zhang2015conditional, Liu2016inference}.
What unifies all of these works is their concern with the information contained in indirect associations, which accounts for the high FDR of isolated pairwise comparisons.
All aim, directly or indirectly, to exclude gene pairs with high MI but low CMI when conditioned on another gene or biological response.
Thematically, what ties together this literature is a methodological wariness of \textbf{redundancy}.
The highly complex dependency structure within GRNs leads to an information landscape characterized by \textbf{redundant information}, or the sharing of equivalent uncertainty among more than two variables.

Separate from this work limiting redundancy, a parallel effort aimed to capture and employ the equally elusive quality of synergy.
Whereas redundancy characterizes indirect associations that ought to be discarded, synergy metrics have been proposed in order to locate multivariate associations that are not reducible to pairwise interactions.
Synergy came into the networks literature from neuroscience, originally from the work of Brennan \etal \cite{Gat1999synergy,Brenner2000synergy}.
Within a signal, neural spikes convey less information in isolation than when they are taken together in their succession.
This formula for synergy was previously defined in the information literature as interaction information \cite{McGill1954multivariate}.
A positive value is said to indicate synergy, while a negative value indicates redundancy.
Anastassiou applied this definition of synergy to gene expression data in \cite{Anastassiou2007synergy}.
To our knowledge, Watksinon and colleagues first utilized this synergy to infer gene networks in \cite{Watkinson2008identification} and \cite{Watkinson2009SACLR}, also in the former's dissertation \cite{Watkinson2011PhD}.

Interaction information, equivalent to Brenner's synergy for three variables, was first introduced by McGill in  \cite{McGill1954multivariate}, and is formalized in contemporary notation\footnote{McGill's original formulation uses the terminology of transmission information in place of mutual information, and was defined on empirical distributions. His definition of interaction information for three variables can be found in Eq.~13-14 on p. 101 of \cite{McGill1954multivariate}. Interaction information is defined for more than three variables, however. Furthermore, up to a difference of sign convention, it is equivalent to the quantity developed by Hu Kuo Ting in \cite{Ting1962amount} and Anthony Bell in \cite{Bell2003co}. Ting defines this alternative quantity as an additive set function, said to quantify multivariate amounts of information, but does not give it a particular name \cite{Ting1962amount}. Bell designates the quantity as co-information. In this work, we discuss only interaction information in the three variable case, taking the sign convention used in \cite{McGill1954multivariate}. This is the setting in which the PID framework is situated \cite{Williams2010}.
}
as an extension of mutual information:
\begin{equation}
    I(T;X;Y) = I(T;X | Y) - I(T;X).
\end{equation}
where $T$ is the received signal and $X$ and $Y$ are the transmitted signals.  Note that the quantity is symmetric in the three arguments, and thus the designation of target signal is arbitrary.
In this definition, the three variables are said to have a positive interaction when knowledge of one variable increases the information flow between the others, i.e. when $I(T;X |Y) > I(T:X)$.
As discussed above, this phenomenon has come to be understood as signifying synergy \cite{Timme2014experimentalist, Williams2010}.
By contrast, a negative interaction occurs when knowledge of one variables reduces the information shared between the other two, i.e. $I(T;X|Y) \leq I(T;X)$.
This is understood as signifying redundancy \cite{Timme2014experimentalist,Williams2010}.

Identifying biological components that act synergistically is a fundamental problem in systems biology.
In molecular biology, to take one example, one is often interested in identifying the synergistic interactions of genes upon a biological response. 
For instance, it is a commonplace that transcription factors behave synergistically as genetic regulators \cite{Wang2007synergistic, Sauer1995synTFs}.
It is not difficult to see the shortcoming of the interaction information approach, which classifies triplets as \textbf{net synergistic} or \textbf{net redundant}.
In truth, synergy and redundancy are not incompatible opposites.
In reality, we expect the two to often co-occur within the same variables.
When there is significant cross-talk and dependency between $X$ and $Y$, then we expect some quantity of redundant information about any significant causal relationship $X \to T$ or $Y \to T$.
On the other hand, we expect, for many collections of components of common interest to the same biological function, some element of synergy to their contribution.
GRNs, for instance, are generally characterized by a good deal of pathway redundancy and robustness, and this becomes especially pronounced in cancer cells; see, for instance, literature on network entropy \cite{Demetrius2005robustness,Manke2006entropic} and on signaling entropy \cite{Teschendorff2014signalling,Teschendorff2015scalefree,teschendorff2017singleCell}. 
Thus, between alternative feedback mechanisms and topological cyclicity involving genes outside $\{ X, Y, T\}$, we do not generally expect a situation where $X \to Y \to T$ is a Markov chain.
To the contrary, the knowledge of both $X$ and $Y$ together will often better inform us regarding activated processes than either alone, providing synergistic information about $T$.
Synergy should estimate that information as an explicit quantity beside, and not in opposition to, redundancy induced by cross-talk between the genes.
The PID framework attempts to disentangle these two quantities of synergy and redundancy \cite{Williams2010, Timme2014experimentalist}, thus offering an attractive framework for the inference of response-specific synergy networks.
We will now review the development of PID and its current state in the literature.

\subsection{Partial Information Decomposition (PID)}
\label{subsection:PreviousWork.PID}

Williams and Beer introduced the PID paradigm in \cite{Williams2010}.
Their work was motivated by interaction information, and aimed to address the conflation of synergy and redundancy within $I(T;X;Y)$.
In that vein, they sought to decompose the information provided about $T$ by predictors $X$ and $Y$ into four atoms of information:
redundancy, synergy, and unique information for each predictor.
The atoms would fulfill the following equations:
\begin{align}
\label{eq:PID_1}
\tag{E1}
I(T;X,Y) &= \underbrace{R(T;X,Y)}_{\text{Redundant Info.}} + \underbrace{S(T;X,Y)}_{\text{Synergistic Info.}} \\
\nonumber & + \underbrace{U_X(T;X,Y)  + U_Y(T;X,Y) }_{\text{Unique Infos.}} \\
\label{eq:PID_2}
\tag{E2}
I(T;X) &= R(T;X,Y) + U_X(T;X,Y)\\
\label{eq:PID_3}
\tag{E3}
I(T;Y) &= R(T;X,Y) + U_Y(T;X,Y)
\end{align}
This particular decomposition is generally referred to as `bivariate', since it includes only two predictor variables \cite{Bertschinger2013shared}.
The framework in \cite{Williams2010} is defined for an arbitrarily large (though finite) number of predictors.

In general, a PID can be uniquely defined by its redundancy function $\Ired$.
An $\Ired$ is an order-preserving function on a lattice of ascending collections of \textbf{sources}, the designation for subsets of predictor variables.
Williams and Beer cited, among others, the information lattice structure from \cite{Bell2003co}, in developing their own lattice formalism in \cite{Williams2010}.
We will not be utilizing this lattice structure in our current work except in passing, as we are restricting ourselves to the bivariate case, with the four PID elements $R, U_X, U_Y, S$ corresponding to the lattice elements $\{X\}\{Y\}, \{X\}, \{Y\},$ and $\{X, Y\}$ respectively \cite{Williams2010}.

If one of these four atoms is defined, then the other three will follow from Eqs~\ref{eq:PID_1}-\ref{eq:PID_3}.
Williams and Beer introduced one measure of redundancy, the $\Imin$ function, often denoted as $I_{\text{min}}$ in the literature,
which we present in Def.~\ref{defn:redundant-info.Imin}.
This redundancy function, in turn, defines all four of the atoms $\Rmin, \Smin, \Umin{X},$ and $\Umin{Y}$.
With their definition, Williams and Beer achieved the aim of decomposing interaction information into synergy and redundancy:
\begin{equation}
    I(T;X;Y) = \Smin(T;X,Y) - \Rmin(T;X,Y)
\end{equation}
Their work opened a novel approach to multivariate information.
It offered a welcome break from the previous perspective on synergy and redundancy, which held them to be mutually exclusive quantities as captured by the sign of interaction information \cite[and above]{Timme2014experimentalist, EntropySIOverview2018}.
This is especially attractive to our interest in biological network analysis, where we expect both redundancy and synergy in any pair of interacting genes.

Moreover, the $\Imin$ redundancy function demonstrated many desirable properties.
These include:
\begin{enumerate}
    \item[0.] \textbf{Nonnegativity.}
    Redundant information is a nonnegative quantity.
    \begin{equation}
        \label{eq:WBaxiom.P}
        \tag{P}
        \Ired(T;\XX_1,...,\XX_n) \geq 0
    \end{equation}
    \item \textbf{Symmetry.} Redundant information is invariant under permutations of sources:
    \begin{equation}
        \label{eq:WBaxiom.Sym}
        \tag{S}
        \Ired(T;\XX_1, ..., \XX_n) = \Ired(T;\XX_{\sigma(1)}, ..., \XX_{\sigma(n)})
    \end{equation}
    \item \textbf{Monotonicity.} Every additional source included can only decrease the amount of redundant information. 
        \begin{equation}
        \label{eq:WBaxiom.M}
        \tag{M}
        \Ired(T; \XX_1, ..., \XX_n) \leq \Ired(T; \XX_1, ..., \XX_{n-1})
        \end{equation}
    \item \textbf{Self-Redundancy.} The redundant information of a single source is the mutual information between that source and the target.
    \begin{equation}
    \tag{SR}
        \label{eq:WBaxiom.SR}
        \Ired(T;\XX) = I(T;\XX)
    \end{equation}
\end{enumerate}
(\ref{eq:WBaxiom.Sym}),(\ref{eq:WBaxiom.M}), and (\ref{eq:WBaxiom.SR}) would later become identified as the Williams-Beer axioms \cite{Bertschinger2013shared,EntropySIOverview2018}.
Every subsequent PID redundancy would satisfy (\ref{eq:WBaxiom.Sym}) and (\ref{eq:WBaxiom.SR}), at a minimum.

However, the $\Imin$ redundancy function was not widely accepted as quantifying redundant information \cite{EntropySIOverview2018}, and multiple alternatives have been proposed \cite{Harder2013bivariate,Bertschinger2014,Ince2017,Finn2018main}.
The most common complaint is that the $\Imin$ function does not distinguish the same information from the same \textit{amount} of information.
Specifically, regarding each target outcome $T=t$,  $\Imin$ compares the \textit{expected} amount of information that each predictor provides about that target event, regardless of the degree of overlap
in elementary events.
Harder et al.\@ \cite{Harder2013bivariate} constructed a particularly striking example of the drawback to this approach  (see also \cite{Griffith2014quantifying}), later dubbed the `two-bit copy' problem \cite{Ince2017,EntropySIOverview2018}.
When $T = (X,Y)$ is a two-bit copy of i.i.d. one-bit binary variables $X$ and $Y$, 
the $\Imin$ function assigns one bit each of redundant and synergistic information to the decomposition of $I(T;X,Y)$.
For any $t=(x,y) \in \{0,1\}^2$, we have that $X$ and $Y$ provide the same amount of information about that outcome, on average.
All the information contained within $I(T;X)$ (and $I(T;Y)$) is assigned to redundancy, and the synergy atom follows, though it is perhaps more intuitive for a PID to designate $R=S=0$ and $U_X = U_Y = 1$.
Beyond this example, a more thorough review of the objections and alternatives to the $\Imin$ PID can be found in the review \cite{EntropySIOverview2018}, the overview of a 2018 special issue of \textit{Entropy} on the PID framework.

Despite the rich development of the PID framework, there are two areas of research that have remained less explored until relatively recently.
First, the literature on the application of PID to non-neural biological network inference has been relatively sparse.
The most visible effort has been in the development by Chan \etal of the PID and context (PIDC) methodology for gene regulatory network (GRN) inference \cite{Chan2017Cell,Chan2018PhD}.
More recently, Cang and Qin applied a similar approach to spatial single-cell sequencing data \cite{Cang2020inferring}.

Second, the PID literature has largely been confined to the study of discrete or discretized variables \cite{EntropySIOverview2018}.
Barrett contributed an early foray into the PID analysis of continuous variables in \cite{Barrett2015}, in which he considered the extension of the PID framework to jointly Gaussian variables, both static and dynamic.
In that context, Barrett demonstrated that --- for a large class of potential PIDs, including those in \cite{Williams2010}, \cite{Griffith2014quantifying}, and \cite{Bertschinger2014}, but excluding \cite{Finn2018main} and \cite{Ince2017} --- the PID of a trivariate jointly Gaussian system must necessarily take the same form, which they term the minimal mutual information (MMI) PID.
Our computation of the $\Imin$ PID for a linear interaction can be viewed as a natural limit of the MMI PID in \cite{Barrett2015}, though
we present a stand-alone proof that aligns with the narrative focus upon the more general interaction kernels $g(X,Y)$ that we will be exploring in this effort. 

Besides the MMI PID for Gaussians, there are other, much more recent works relevant to extending PID to continuous variables.
Ari Pakman et al. has extended the $\Ibroja$ PID (\cite{Bertschinger2014}) to continuous variables \cite{Pakman2021estimating}.
Kyle Schick-Poland et al. \cite{Schick2021partial} developed a signed continuous PID in \cite{Schick2021partial}, based upon the pointwise differentiable $\isx$ PID developed by Makkeh et al. in \cite{Makkeh2021differentiable}.
This PID shares much in common with the $\Ipm$ PID, including a pointwise approach to decomposing MI as well as the separation of the signed PID lattice into positive and negative sublattices, each nonnegative for discrete variables.

The contribution of our work, then, is two-fold.
First, we explore the potential of PID as a tool for the inference of synergy networks, particularly the PM PID of \cite{Finn2018main} (see Definition \ref{defn:redundant-info.Ipm} for the relevant definition of PID in this setting), as well as the original $\Imin$ PID \cite{Williams2010} (Def.~\ref{defn:redundant-info.Imin}).
The $\Ipm$ PID is unusual in that it allows for signed information atoms, as do the $\Iccs$ and $\Isx$ PIDs.
We explore the implications of this in both simulation and analysis, and find that it allows for counter-intuitive behavior that undermines the specificity of the synergy defined by the $\Ipm$ PID.
We see this in simulation, and also analytically for a degenerate system of random variables $(X,Y,T)$, in which $T$ is a deterministic response --- that is, $T=g(X,Y)$ --- of the jointly Gaussian predictors $X$ and $Y$.
Our second contribution is to be the first work to extend both the $\Imin$ and $\Ipm$ PIDs to such a system, and to the continuous setting writ large.
We highlight in Theorem~\ref{theorem:GenericKernel.UniqueInfo} of Section~\ref{section:GenericResults}, that we are able to provide an explicit integral formula for the unique information atoms $U_X$ and $U_Y$ in such systems for a general class of kernels $g$.
This formula depends upon the relative magnitude of the partial derivatives of $g$, connecting the analytic properties of $g$ to the probabilistic properties of $(X,Y,T)$.
We note that our work parallels (and importantly diverges from) the recent work done by Pakman \etal in \cite{Pakman2021estimating} for the $\Ibroja$ PID \cite{Bertschinger2014}, in that both our work and theirs is concerned with computing the unique information atom for continuous variables.
There is strong resemblance between some of the noise-free interaction that we are considering (Eq.~\ref{eq:sigSw}) and the example neural models (III.1) and (III.2) in \cite{Pakman2021estimating}.
Besides considering different PIDs ($\Imin$ and $\Ipm$ rather than $\Ibroja$), our approach, focus, and conclusions are also distinct.

\section{Preliminaries}
\label{section:preliminaries}

In this section are provided the prerequisites from probability, information theory, and partial information.
We will also take this opportunity to clarify all the notation to be used in this work.
We also introduce a shorthand for non-discrete variables (marginally continuous versus jointly continuous in Def.~\ref{defn:randomSources}) in order to positively differentiate non-degenerate continuous sources from degenerate sources.

\subsection{Random Variables and Random Sources}
\label{subsection:preliminaries.RandomSources}

Let $(\Omega, \mathcal{B}, \mu)$ be our assumed probability space. A random variable $X$ is a measurable mapping $\Omega \to \mathcal{A}_X$, where $\mathcal{A}_X$ is the alphabet of $X$.
In this paper, all alphabets are subsets of $\mathbb{R}^k$ for some $k$.
We say that $X$ is a discrete variable if $\mathcal{A}_X$ is countable.
We say that $X$ is continuous if it admits a density $p_X$ on $\mathbb{R}$.
More precisely, we say that $X$ is continuous if it induces a measure $\mu_X$ on $\mathbb{R}$ that is absolutely continuous with respect to Lebesgue measure $\lambda$ (denoted $\mu_X \ll \lambda$ ), in which case the density is the Radon-Nikodym derivative $p_X = \frac{d \mu_X}{d \lambda}$.

Following \cite{Williams2010}, we define a \textbf{random source} to be a finite collection of random variables.
One can think of a random source $\XX = \{X_1, ..., X_N \}$, with associated alphabets $\mathcal{A}_{X_1}, ..., \mathcal{A}_{X_N}$ as a random variable $\XX$ on the product alphabet $\mathcal{A}_{\XX} = \otimes_i \mathcal{A}_{X_i}$. Nonetheless, we distinguish the two concepts for clarity, and reserve `random variable' to refer to real-valued variables.
Moreover, we may elect to treat a random source $\XX = \{ X_i \}_i$ as a random vector $\XX = (X_i)_i$ when convenient, e.g. when we want to admit a density $p_{\XX}$ on $\mathbb{R}^k$
\footnote{We implicitly assume throughout that every function of random sources that makes use of their joint densities is invariant under permutation of labels. For instance, when we define entropy $H(\{X_1, ..., X_k\})$, we use the density $p_{X_1, ..., X_k}(x_1, ..., x_k)$, but the entropy would be the same if we used $p_{X_{\sigma(1)}, ..., X_{\sigma(k)}}(x_{\sigma(1)}, ..., x_{\sigma(k)})$ instead. }.
We similarly use $(\mathcal{A}_{\XX}, \mu_{\XX})$ to denote the induced probability space in $\mathbb{R}^k$.
We will denote the image of an event $E \subset \Omega$ in $\mathcal{A}_{\XX}$ as $\pi_{\XX}(E)$, where $\pi_{\XX}$ is here synonymous with the random vector $\XX$ as a mapping $\Omega \to \mathcal{A}_X$.
Likewise, for $U \subset \mathcal{A}_{\XX}$ we may then use $\pi_{\XX}^{-1} U$ to represent the pullback of the event $\{ \XX \in U \}$.

\begin{defn}[Random Sources]
\label{defn:randomSources}
A finite collection of random variables $\XX = \{X_1, ..., X_k\}$ on the space $(\Omega, \mu)$ is called a \textbf{random source}.
\begin{itemize}
    \item If every variable $X_i$ is discrete, we say that $\XX$ is a discrete source.
    \item If every variable $X_i$ is continuous, we say that $\XX$ is \textbf{marginally continuous}.
    \item We say that $\XX$ is \textbf{(jointly) continuous} if the joint distribution of $X_1, ..., X_k$ admits a joint density $p_{\XX}$. Equivalently, $\XX$ is jointly continuous if it induces a measure $\mu_{\XX}$ on $\mathbb{R}^k$ such that $\mu_{\XX} \ll \lambda^k$, in which case $p_{\XX} = \frac{d \mu_{\XX}}{d \lambda^k}$.
    \item We say that $\XX$ is \textbf{non-degenerate} if it is discrete or continuous. It is \textbf{degenerate} otherwise.
\end{itemize}
\end{defn}

For the density of a continuous source, we may represent it as $p(\XX) := p_{\XX}(\XX)$ when the subscript is obvious from context.
We may still allow singular sets of measure zero that are possible, in the sense that arbitrary neighborhoods have positive probability
\footnote{For instance, in when considering a noise-free interaction $T=g(X,Y)$ for kernel (\ref{eq:sigSw}), we will have that the random source $\{ Y, T \}$ admits a density $p_{Y,T}$ almost everywhere, but not when $Y=0$, an event with zero probability.
Nonetheless, any neighborhood of $(Y,T) = (0,0)$--- i.e. any event $E=\pi_{Y,T}^{-1} B_{\epsilon}(0,0)$, where $B_\epsilon$ denotes an open $\epsilon$-disk --- will have positive probability.
Thus, we include $(0,0) \in \mathcal{A}_{X,T}$.}.

We highlight a few immediate properties.
\begin{proposition}[Remarks on Continuous Sources]
For random sources $\XX, \bm{Y}$, we have the following
\begin{enumerate}
    \item If $\XX$ is continuous, then it is marginally continuous.
    \item If $\XX$ is jointly (respectively, marginally) continuous, then every subset $\XX' \subset \XX$ is also jointly (respectively, marginally) continuous.
    \item If $\XX$ and $\bm{Y}$ are both marginally continuous, then so is $\XX \cup \bm{Y}$.
    \item If $\XX$ and $\bm{Y}$ are both continuous, it is not necessary that $\XX \cup \bm{Y}$ be continuous.
\end{enumerate}
\end{proposition}
The first two properties follow from the fact that the integration of a continuous function is continuous.
The third property is trivial.
The final property is demonstrated by the following counterexample.

\begin{example}[Non-Degenerate and Degenerate Bivariate Gaussian]
\label{example:infoTheory.bivariateGaussian}
Suppose $X, Y \sim N(\bm{0}, \Sigma)$ where
$$\Sigma = \begin{bmatrix}
  \sigma^2 & \sigma^2 \rho\\
  \sigma^2 \rho & \sigma^2
\end{bmatrix}.$$
Then $\{X, Y\}$ is a continuous source if and only if $|\rho| < 1$. If we allow $\rho = \pm 1$, we have identical Gaussians with the joint distribution $p_{X,Y}(x,y) = \frac{1}{\sqrt{2 \pi}} \delta(y \mp x) e^{-x^2/2}$, which is a distribution (generalized function) and not a density.
We still have that $\{ X \}$ and $\{ Y \}$ are continuous sources, but $\{ X, Y \}$ is only marginally continuous and degenerate.
\end{example}

Finally, we will often need to compute expectations by switching between continuous sources and alphabets that are not jointly continuous (e.g. $\mathcal{A}_{X,T}$ and $\mathcal{A}_{X,Y}$ for noise-free interactions in Sec.~\ref{section:GenericResults}).
We use the following notation for almost-sure equality of random variables:
\begin{align}
f(\XX) &\oeq g(\YY) \text{ if}\\
f(\XX(\omega)) &= g(\YY(\omega)) \text{ for $\mu$-a.e. } \omega \in \Omega. \nonumber
\end{align}
This refers to the essential equality of the pullback functions $\pi^{\star}_{\XX} f$ and $\pi^{\star}_{\YY} g$ on $\Omega$.
Almost-sure equality is more than sufficient for equality in expectation:
\begin{align}
f(\XX) \oeq g(\YY) \Longrightarrow \mathbb{E} f(\XX) = \mathbb{E} g(\YY).
\end{align}

We may now define conditional sources for non-degenerate distributions, using their mass and density functions.
\begin{defn}[Conditional Random Sources]
Let $\XX, \YY$ be random sources, either both discrete or jointly continuous. Let $\yy \in \mathcal{A}_{\YY}$ such that
$p_{\YY}(\yy) > 0$.

The conditional source $\XX$ given $\YY = \yy$ is denoted $\XX_{|\YY = \yy} = \XX_{|\yy}$, and is the discrete (continuous) random vector with alphabet $\mathcal{A}_{\XX | \yy} \subset \mathcal{A}_{\XX}$, with mass function (density):
\begin{align}
    p(\xx_i | \yy_j) &:= \frac{p(\xx_i, \yy_j)}{p(\yy_j)} &\text{ (discrete sources)}\\
    \nonumber
    p_{\XX | \YY} (\xx | \yy) &:= \frac{p_{\XX,\YY} (\xx, \yy)}{p_{\YY}(\yy)}
    &\text{(continuous sources)}.
\end{align}
\end{defn}

Pointwise information of a random variable is quantified by a logarithmic transformation of the of the distribution functions, typically referred to as the surprisal.
\begin{defn}
Let $\XX$ be a discrete or continuous random source.
Then the pointwise surprisal $s_{\XX}$ is defined as the function
\begin{align}
   \label{eq:defn:surprisal}
    s_{\XX}(\xx)
    &=
    \begin{cases}
    \log \frac{1}{p(\xx_i)}  &\text{ (discrete)}\\
    \log \frac{1}{p_{\XX}(\xx)}
    &\text{ (continuous)}
\end{cases}
\end{align}
For discrete sources, $s_{\XX}$ takes $\mathcal{A}_{\XX}$ as its domain.
For continuous sources, $s_{\XX}$ is defined where $\mu_{\XX}$ admits a nonzero density $p_{\XX}$.

If $\XX$ and $\YY$ are both discrete or jointly continuous, we define the conditional surprisal of $\xx$ given $\YY = \yy$ as
\begin{align}
   \label{eq:defn:surprisal.conditional}
    s_{\XX|\YY}(\xx | \yy)
    &=
    \begin{cases}
    \log \frac{1}{p(\xx_i | \yy_j)}  &\text{ (discrete)}\\
    \log \frac{1}{p_{\XX | \YY}(\xx | \yy)}
    &\text{ (continuous)}
\end{cases}
\end{align}
\end{defn}

Surprisal measures the unlikelihood of a particular outcome.
The conditional surprisal is the difference between the joint and individual surprisals, that is,
\begin{equation}
    s_{\XX | \YY} = s(\xx, \yy) - s(\xx).
\end{equation}

\subsection{Information Theory}
\label{subsection:preliminaries.InfoTheory}

The Shannon entropy, originally defined in \cite{Shannon1948}, quantifies the uncertainty in the value of a variable $X$.
We assume the convention whereby discrete and differential entropy are distinguished as $H$ and $h$, respectively, due to the incommensurability of the two notions, as seen when sequences of continuous (discrete) variables converge to a discrete (continuous) distribution \cite[Sec. 8.3]{CoverThomas2005}.

\begin{defn}[Shannon Entropy for Random Sources]
\label{defn:entropy}
Let $\XX$ be a random source, either discrete or continuous. Then the entropy of $\XX$ is defined as:
\begin{align}
    \label{eqn:defn.entropy.discrete.source}
    H(\XX) &= \mathbb{E} s_{\XX}(\xx) \\
    \nonumber & = - \sum p(\XX_i) \log p(\XX_i) 
    &\text{ (discrete)}\\
    h(\xx) &= \mathbb{E} s_{\XX}(\xx) \\
  \nonumber & = - \int p_{\XX}(\xx) \log p_{\XX}(\xx) d\bm{x}
    &\text{ (continuous)}
\end{align}

If $\YY$ is another source, such that either both $\XX$ and $\YY$ are discrete or $\XX \cup \YY$ is continuous, then the conditional entropy of  $\XX$ given $\YY$ is
\begin{align}
    \label{eqn:defn.entropy.conditional.discrete.source}
    H(\XX | \YY) &= \mathbb{E} s_{\XX| \YY}(\xx|\yy) \\
    \nonumber &= - \sum_{i,j} p(\bm{x}_i, \bm{y}_j) \log \frac{p(\bm{x}_i, \bm{y}_j)}{p(\bm{y}_j)}
    \\
    h(\xx | \yy) &= \mathbb{E} s_{\XX| \YY}(\xx|\yy) \\
  &= \nonumber - \int p_{\bm{X},\bm{Y}}(\bm{x}, \bm{y}) \log \frac{p_{\bm{X},\bm{Y}}(\bm{x}, \bm{y})}{p_{\bm{Y}}(\bm{y})} d(\bm{x},\bm{y})
\end{align}

\end{defn}

We may now define mutual information (MI) and KL divergence for sources that are either discrete or sufficiently continuous.
If entropy quantifies the uncertainty of a source $\XX$, then we may say that MI quantifies the shared uncertainty between $\XX$ and $\YY$.
Equivalently, MI is the reduction of uncertainty possible in one source from knowledge of the other.

\begin{defn}[Mutual Information for Random Sources]
\label{defn:MI}
Let $\bm{X}$ and $\bm{Y}$ be random sources such that either both are discrete or $\bm{X} \cup \bm{Y}$ is continuous. Then the mutual information of $\bm{X}$ and $\bm{Y}$ is given by
\begin{subequations}
\begin{align}
    \label{eq:MI.defn}
    I(\XX; \YY)
    &= \mathbb{E} \log \frac{ s_{\XX}(\XX)s_{\YY}(\YY)  }{ s_{\XX, \YY}(\XX,\YY) }\\
    \label{eq:MI.entropicDifference}
    &=
    \begin{cases}
    H(\XX) - H(\XX|\YY)  &\text{ (discrete)}\\
    h(\XX) - h(\XX|\YY)
    &\text{ (continuous)}
    \end{cases}
\end{align}
\end{subequations}
If $\ZZ$ is another source and either $\XX,\YY,\ZZ$ are all discrete or $\XX \cup \YY \cup \ZZ$ is continuous, then the conditional mutual information of $\XX$ and $\YY$ given $\ZZ$ is
\begin{subequations}
\begin{align}
    I(\XX; \YY | \ZZ) 
    &= \mathbb{E} \log \frac{ s_{\XX|\ZZ}(\XX|\ZZ)s_{\YY|\ZZ}(\YY|\ZZ)  }{ s_{\XX, \YY|\ZZ}(\XX,\YY|\ZZ) }\\
     &=
    \begin{cases}
    H(\XX|\ZZ) - H(\XX|\YY,\ZZ)  &\text{ (discrete)}\\
    h(\XX|\ZZ) - h(\XX|\YY,\ZZ)
    &\text{ (continuous)}
    \end{cases}
\end{align}
\end{subequations}
\end{defn}

Unlike entropy, we do not distinguish between discrete and continuous MI in notation.
This convention is justified as both the discrete and continuous definition presented in Def.~\ref{defn:MI} are special cases of the master definition in Def.~\ref{defn:MI.general}.

Kullback-Liebler (KL) divergence is colloquially described as quantifying the entropic `distance' between two distributions.
From an early application of Shannon information theory to statistics, it was originally conceptualized as a measure of the `information for discrimination' between competing hypotheses \cite{Kullback1951}.
\begin{defn}[Kullback-Liebler Divergence]
\label{defn:KL_divergence}
Suppose $\XX$ and $\YY$ are random sources, either both discrete or individually continuous, on the same alphabet $\mathcal{A} = \mathcal{A}_{\XX} = \mathcal{A}_{\YY}$.
Then the Kullback-Liebler (KL) divergence between their distributions
is given by
\begin{align}
    D(p_{\XX} | p_{\YY})
    &= \mathbb{E} \left[ s_{\YY}(\XX) - s_{\XX}(\XX) \right] \\
    \nonumber &=
    \begin{cases}
     \sum_{\xx_i \in \mathcal{A}} p_{\XX}(\xx_i) \log \frac{p_{\XX}(\xx_i)}{p_{\YY}(\xx_i)} &\text{ (discrete)}\\
     \int p_{\XX}(\xx) \log \frac{p_{\XX}(\xx)}{p_{\YY}(\xx)} d\xx
    &\text{ (continuous)}
    \end{cases}
\end{align}
\end{defn}
KL divergence is not a true distance or even semimetric, as it is not symmetric and does not obey the triangle inequality.
Moreover, it is infinite when the first distribution is not absolutely continuous with respect to the second.


We have covered the usual definitions from information theory.
However, since we will be working with degenerate distributions in Section~\ref{section:GenericResults}, we will briefly outline the more general framework for information for variables of arbitrary measure.
In particular, for continuous random variables, we have that
\begin{equation}
    I(T;X,Y) = \infty
\end{equation}
when $T$ is determined by $X$ and $Y$.

For any alphabet $\mathcal{A}$, let $\Lambda(\mathcal{A})$ be the collection of all countable partitions of $\mathcal{A}$.
For a continuous variable $\XX$ and partition $\text{\emph{P}} \in \Lambda(\mathcal{A}_{\text{\emph{P}}})$, let  $[\XX]_{\text{\emph{P}}}$ be the discretized variable:
$$
\forall P_i \in P, \;
[\XX]_{\text{\emph{P}}}(\omega) = i 
\Leftrightarrow
\XX (\omega) \in P_i.
$$

\begin{defn}
\label{defn:MI.general}
Let $\XX$ and $\YY$ be two random sources, with alphabets $\mathcal{A}_{\XX}$ and $\mathcal{A}_{\YY}$.
The mutual information between $\XX$ and $\YY$ is given by
\begin{equation}
\label{eq:MI.general}
    I(\XX; \YY) = \sup_{\text{\emph{P}},\text{\emph{Q}}} I([\XX]_\text{\emph{P}}; [\YY]_{\text{\emph{Q}}})
\end{equation}
where the supremum is taken over all discretizations $(P,Q) \in \Lambda(\mathcal{A}_{\XX}) \times \Lambda(\mathcal{A}_{\YY})$.
\end{defn}
The supremum in Eq.~(\ref{eq:MI.general}) is monotonically increasing for successively finer partitions $\text{\emph{P}}$ and $\text{\emph{Q}}$.
Thus, this definition can be approached with standard methods from real analysis and measure theory.
For our purposes, we use this definition only to demonstrate the following proposition.
\begin{proposition}
\label{prop:infiniteMI}
Let $\XX, \YY$ be two random sources, each individually continuous, but such that $\XX \cup \YY$ is not continuous. Then $I(\XX;\YY) = \infty$.
\end{proposition}
The proof is conceptually identical to that of Lemma 7.4 in \cite{Gray2011entropy}.
\begin{proof}
Let $n = |\XX|$ (the dimension of $\XX$), and $m = |\YY|$. Since $\XX$ and $\YY$ are each continuous, their induced measures $\mu_\XX$ and $\mu_\YY$ are absolutely continuous with respect to Lebesgue measure $\lambda^n$ and $\lambda^m$. Thus, $\mu_\XX \times \mu_\YY \ll \lambda^{n+m}$. If $\mu_{\XX, \YY} \ll \mu_\XX \times \mu_\YY$, then we would have $\mu_{\XX,\YY} \ll \lambda^{n+m}$, which would imply $\XX \cup \YY$ continuous. Thus, $\mu_{\XX, \YY} \not\ll \mu_\XX \times \mu_\YY$.

This implies there is a set $P \subset \mathbb{R}^{n+m}$ such that
$$
\mu_{\XX, \YY}(P) > 0 \text{ and } (\mu_\XX \times \mu_\YY(P)) = 0 .
$$ For the partition $\text{\emph{P}} = \{P_0, P_0^C \}$, we have that $I(\XX_{\text{\emph{P}}};\YY_{\text{\emph{P}}}) = \infty$, and our result follows from Def.~\ref{defn:MI.general}.
\end{proof}

\begin{corollary}
\label{corollary:infiniteMI.sumOfDiracs}
Let $\XX, \YY$ be two random sources with joint distribution
\begin{equation}
    p_{\XX,\YY}(\xx, \yy) = \sum_{i=1}^M \delta(g_i(\xx, \yy)) p_{\XX}(\xx)
\end{equation}
where $\{g_i\}$ is a collection of real-valued functions with isolated zeros. Then
\begin{equation}
    I(\XX, \YY) = \infty.
\end{equation}
\end{corollary}

In such a situation, we may say that $\XX$ provides \textbf{perfect} (or \textbf{infinite}) \textbf{information} information about $\YY$.
We will use the two words, perfect and infinite, interchangeably, depending on whether we want to emphasize either the collapse of uncertainty or the quantitative irregularity (i.e. infinite-valued information functions).

This concludes our review of information theory.
We may now introduce our extension of the $\Imin$ and $\Ipm$ PIDs to continuous interactions of random variables.

\subsection{Partial Information Decomposition}
\label{subsection:preliminaries.PID}

One of the aims of this work is to extend the partial information framework of Williams and Beer \cite{Williams2010}, and the follow-up work by Finn and Lizier in \cite{Finn2018main}, to interactions of continuous variables.
To that end, we define the $\Imin$ and $\Ipm$ PIDs for both discrete variables and variables each jointly continuous to the target.
Our definitions are straight-forward extensions of those in \cite{Williams2010} and \cite{Finn2018main}.
The main contribution of our work towards a theory of continuous PID is to formally define continuous PID's and analytically explore the implications of continuity on the (often counter-intuitive) distribution of the information atoms (which is further reinforced by simulation), as developed in Sec.~\ref{section:GenericResults}.
We will ultimately connect the the information-theoretic dependencies quantified in the PID framework to the analytic dependency of a target variable upon its predictors --- that is, to the partial derivatives within a deterministic context.

Recall from Eqs.~(\ref{eq:PID_1}-\ref{eq:PID_3}) that, given target $T$ and predictors $X$ and $Y$, the bivariate PID is the decomposition of $I(T;X,Y)$ into the quadruple
\begin{equation}
I(T;X,Y) = R + U_X + U_Y + S.
\end{equation}
We denote this tuple $\mathfrak{d}$.
Every PID is uniquely defined by its redundancy function.
For a given redundancy function $\Ired^{\alpha}$, labeled by the string symbol $\alpha$, we denote the corresponding bivariate PID
\begin{align}
    \mathfrak{d}_\alpha &= (R^{\alpha}, S^{\alpha}, U_X^{\alpha}, U_Y^{\alpha}), 
\end{align}
where
\begin{align}
    R^{\alpha} &= I_{\cap}^{\alpha}(T;X,Y),\\
    U_X^{\alpha} &= I(T;X) - I_{\cap}^{\alpha}(T;X,Y),\\
    U_Y^{\alpha} &= I(T;Y) - I_{\cap}^{\alpha}(T;X,Y),\\
    S^{\alpha} &= I(T;X,Y) - I(T;X) - I(T;Y) + I_{\cap}^{\alpha}(T;X,Y).
\end{align}
In this work, we focus upon two choices for the redundancy function: $\Imin$ and $\Ipm$.
We will also consider the $\Ibroja$ and $\Iccs$ PIDs in Sec.~\ref{section:SynergyNetworks}.

Williams and Beer defined the redundancy function $I_{\cap}^{\min}$ in \cite{Williams2010}.
We expand their definition to include sources jointly continuous with the target variable.
We will first define specific information, a predictor-specific function of the target variable used to define $\Imin$ in \cite{Williams2010}.
Specific information will serve to assign mutual information from a predictor source $\XX$ to each target outcome $\{T=t\}$.

\begin{defn}[Specific Information]
\label{defn:specificInfo}
Let $\XX$ and $\YY$ be discrete or jointly continuous. Then for any $\yy \in \mathcal{A}_{\YY}$ such that $p_{\YY}(\yy) > 0$, we define the specific information of $\XX$ about $\YY=\yy$ to be
\begin{align}
    I_{\XX}(\yy)
     &= D \left( \, p_{\XX|\YY}(\XX|\yy) \; || \; p_{\XX}(\XX) \right).
\end{align}
Moreover, we let $I_{\XX}(\yy) = 0$ whenever $p_{\YY}(y) = 0$.
\end{defn}

Considered generally, for two sources $\XX$ and $\YY$, specific information is an intermediate integral between mutual information --- defined in expectation of both sources --- and the pointwise mutual information function  $\log \frac{ p(\xx,\yy)  }{ p(\xx) p(\yy) }$ --- defined pointwise on the joint alphabet $\mathcal{A}_{\XX \times \YY}$. 
Formally, specific information is always pointwise function of one of the two sources, taken in expectation of the other.
In the context of PID, it is defined pointwise on the target variable $T$, as an expectation on each predictor source $\XX_k$.

\begin{defn}[Definition of $I_{\cap}^{\min}$]
\label{defn:redundant-info.Imin}
Let $T$ be a target variable and $\XX_1, ... \XX_m$ be a collection of sources. Suppose further either that all are discrete, or each $X_i$ is jointly continuous with $T$, i.e. $\XX_i \cup \{ T \}$ is continuous for each $i$. Then the redundant information provided by the sources $\{ \XX_i \}$ about $T$ is given by:
\begin{equation}
    I_{\cap}^{\min}(T; \XX_1, ..., \XX_m) = \mathbb{E} \min_k I_{\XX_k}(T) .
\end{equation}
\end{defn}

Intuitively, we can be think about this definition of redundant information in the following way.
At every pointwise instantiation of $T$, we consider how much uncertainty in $T$ is eliminated, locally near a specific outcome $T=t$, by each $\XX_i$ individually.
Then, when we take the minimum, we take the value corresponding to the minimal reduction in local uncertainty possible when learning one of the sources.
As discussed earlier in Sec.~\ref{subsection:PreviousWork.PID}
and in the literature \cite{Bertschinger2013shared,Finn2018main,EntropySIOverview2018},
this approach conflates variables containing the same information about a target outcome with the same \textit{amount} of information about that outcome.
This is clearly seen in the `two-bit copy' problem (see Sec.~\ref{subsection:PreviousWork.PID}).

From Def.~\ref{defn:redundant-info.Imin} and Eqs.~\ref{eq:PID_1}-\ref{eq:PID_3}, the definitions of the other atoms of the bivariate $\Imin$ PID $\dmin$ follow: $\Umin{X}, \Umin{Y},$ and $\Smin$.
The $\Imin$ PID fulfills all three Williams-Beer PID axioms discussed previously.
\begin{proposition}
\label{prop:Imin.basicProperties}
The redundancy function $\Imin$ fulfills the properties (\ref{eq:WBaxiom.P}), (\ref{eq:WBaxiom.Sym}),(\ref{eq:WBaxiom.SR}), and (\ref{eq:WBaxiom.M}) in Sec.~\ref{subsection:PreviousWork.PID}.
Moreover, in the continuous bivariate setting $\Rmin$, $\Umin{X}$, and $\Umin{Y}$ are all
nonnegative.
In the discrete bivariate setting $\Rmin$, $\Umin{X}$, $\Umin{Y}$, and $\Smin$ are all
nonnegative.
\end{proposition}
The proofs of these properties in the discrete case can be found in the Appendix of \cite{Williams2010}.
The only subtlety in our case in that we are working with continuous variables.
\begin{proof}
Self-redundancy (\ref{eq:WBaxiom.SR}) follows from Def.~\ref{defn:redundant-info.Imin} and Def.~\ref{defn:MI}.
For non-negativity (\ref{eq:WBaxiom.P}), we note that since $I_{\XX}(t)$ is a KL-divergence, it is nonnegative on its domain. 
Thus, $\min_k I_{\XX_k}(t)$ is likewise non-negative, and nonnegativity of $\Rmin(T; X,Y)$ follows.
Monotonicity (\ref{eq:WBaxiom.M}) follows from pointwise monotonicity:
$$
\min_{k=1,..., m}(I_{\XX_k}(t)) \leq \min_{k=1,..., m-1}(I_{\XX_k}(t)) .
$$
In the bivariate setting, the nonnegativity of $U_X^{\text{min}}$ (and analogously, that of $U_Y^{\text{min}}$) follows from the monotonicity that we have just demonstrated:
\begin{align*}
    \Rmin(T;X,Y) = \Imin(T;X,Y) \leq \Imin(T;X) = I(T;X).
\end{align*}
\end{proof}

We turn now to the other PID considered in this work: the $\Ipm$ PID from \cite{Finn2018main}.
Consider the following decompositions of pointwise MI into surprisals:
\begin{align}
\label{eq:preliminaries.PID.surprisals.PM.1}
    \log \frac{p(\xx,t)}{p(\xx)p(t)}
    &= \log \frac{1}{p(t)} - \log \frac{1}{p(t | \xx)}\\
    \label{eq:preliminaries.PID.surprisals.PM.2}
    &= \log \frac{1}{p(\xx)} - \log \frac{1}{p(\xx | t)}
\end{align}
Such decompositions are the pointwise analogue of the relationship between MI and (conditional) entropy as in (\ref{eq:MI.entropicDifference}).

Whereas the surprisal functions (i.e.~pointwise entropy) are always nonnegative for discrete sources, pointwise mutual information can be negative even for discrete sources, if $p(t|\xx) < p(t)$ (equivalently, $p(\xx|t) < p(\xx)$).
Using the decomposition in Eq.~(\ref{eq:preliminaries.PID.surprisals.PM.2}), Finn and Lizier \cite{Finn2018main} provided an alternative definition of redundant information that relies upon the pointwise axiomatic approach developed in \cite{Finn2018axioms}.
In their definition, redundancy minimizes \textit{each entropic component separately}.
This has the consequence of allowing the signed nature of pointwise mutual information to carry over to their PID atoms in expectation.
Their framework allows each information atom in any PID lattice to be decomposed into the difference of two components, which they term \textbf{specificity} and \textbf{ambiguity}.
We first define their redundancy function, $\Ipm$. \\

\begin{defn}
\label{defn:redundant-info.Ipm}
Let $T$ be a target variable and $\XX_1, ... \XX_m$ be a collection of sources, such that the conditions of Definition~\ref{defn:redundant-info.Imin} hold.
Then the redundant information provided by the sources $\{ \XX_i \}$ about $T$ is defined as the difference of the redundant \textbf{specificity} $I_{\cap}^{\PM,+}$ and the redundant \textbf{ambiguity} $I_{\cap}^{\PM,-}$:
\begin{align}
\Ipm(T; \XX_1, ..., \XX_m)
&= \Ipmp(T; \XX_1, ..., \XX_m) \\ \nonumber & - \Ipmm(T; \XX_1, ..., \XX_m), 
\intertext{where}
    \Ipmp &:=
    \mathbb{E} \min_{k=1,...,m} s_{\XX_k}(\XX_k), 
\\
\Ipmm &:=
    \mathbb{E} \min_{k=1,...,m} s_{\XX_k|T}(\XX_k|T). 
\end{align}
\end{defn}

The $\Ipm$ PID allows one to decompose every atom of information from our bivariate PID $\mathfrak{d}$ into its specificity and redundancy:
\begin{subequations}
\begin{align}
     \label{eq:redundant-info.Ipm:eqn:PID.R}
    R^{\PM} &= R^{+} - R^{-} \\
    \label{eq:redundant-info.Ipm:eqn:PID.Ux}
    U_X^{\PM} &= U^{+} - U^{-} \\
    \label{eq:redundant-info.Ipm:eqn:PID.Uy}
    U_Y^{\PM} &= U^{+} - U^{-} \\
    \label{eq:redundant-info.Ipm:eqn:PID.S}
    S^{\PM} &= S^{+} - S^{-} \\
\intertext{where} 
    \label{eq:redundant-info.Ipm:eqn:sublattice.R}
    R^{\pm}(T; X,Y) &= \Ipmpm(T; X,Y) \\
    \label{eq:redundant-info.Ipm:eqn:sublattice.Ux}
    U^{\pm}_X(T; X,Y) &= H^{\pm}(T; X) - \Ipmpm(T; X,Y)\\
    \label{eq:redundant-info.Ipm:eqn:sublattice.Uy}
    U^{\pm}_Y(T; X,Y) &= H^{\pm}(T; Y) - \Ipmpm(T; X,Y)\\
    \notag
    S^{\pm}(T; X,Y)
    &= H^{\pm}(T; X,Y) - H^{\pm}(T; X)\\
    \label{eq:redundant-info.Ipm:eqn:sublattice.S}
    &\hspace{3mm}- H^{\pm}(T; Y) +  \Ipmpm(T; X,Y)
    \intertext{where we are using $H^{\pm}$ to denote the entropic components:}
    H^{+}(T; \XX) &= H(\XX),\\
    H^{-}(T; \XX) &= H(\XX|T).
\end{align}
\end{subequations}

As we will see, by decoupling the entropic components in mutual information, this latter PID definition allows for an unusual emphasis on the negative element in mutual information and its decomposition, as it occurs in conditionally independent or near-independent predictors.


Further, although the disentanglement between redundant and synergistic information is recognized as a major contribution of the PID framework to information theory, it is well-known
that given any two prospective PIDs of the same variables, we have the following conservation between the four bivariate atoms:
\begin{align}
    \label{eq:conservationBivariatePID}
    \Delta R = \Delta S = - \Delta U_X = -\Delta U_Y.
\end{align}
In other words, any inter-PID gain in synergy is a gain in redundancy, and a loss in unique information.
Thus, for our investigation, the differences in redundant and unique information between the $\Imin$ and $\Ipm$ PIDs correspond to the difference in synergy.

\section{PID-based Inference of Synergistic Networks or Synergy Network Inference}
\label{section:SynergyNetworks}

\subsection{Interaction Network}
\label{subsection:interaction_network}

Motivating our investigations of PID is the task of inferring a synergy network from high-dimensional predictor data, coupled to a target or response variable.
We have in mind, for instance, the goal of developing tools for edge nomination in cancer biology applications.
To create clinically useful models of drug action and sensitivity, it is important to identify key molecular agents and interactions, and to distinguish them from auxiliary pathways unrelated to phenomena of interest.
The PID extension to information theory offers an agnostic, non-parametric approach to identifying the most informative, synergistic combinations of predictor variables.
Consider, for example, a drug response metric (e.g. the half inhibitory concentration $\text{IC}_{50}$) as a target, and the normalized expression levels of active genes as candidate predictors.
The PID framework then gives us a means of identifying and coupling \textit{synergistic} biomarkers while eliminating redundant pathway information.

In this section, we will be working with a simple Gaussian model of biological interaction networks.
\begin{defn}[Interaction Network]
\label{defn:gene_interaction_network}
Let $(\mathcal{V}, \mathcal{E})$ be an undirected graph of $n = |\mathcal{V}|$ \textbf{nodes}, and $\sgn: \mathcal{E} \to \{ \pm 1 \}$ an edge attribute signifying positive or negative regulation. Let $\bm{X} \sim N(\bm{0}, \Sigma)$ be an $n$-dimensional Gaussian vector,  where each $X_i$ signifies the \textbf{activity of node i}. For a given constant $\rho \in (0,1)$, $\bm{X}$ has covariance structure:
\begin{equation}
    \Sigma_{i,j} =
    \begin{cases}
    1, & i=j\\
    0, & \{i,j\} \not\in \mathcal{E}\\
    \pm \rho, & \sgn(\{i,j\}) = \pm 1
    \end{cases}
\end{equation}
Let $\mathcal{E}' \subset \{ (i,j) | \{i,j\} \in \mathcal{E} \}$ be a subset of directed edges, called \textbf{interactions}. 
The \textbf{response} $T$ is the real-valued variable of the form
\begin{equation}
\label{eq:networkMotivation.response}
T = \sum_{(i,j) \in \mathcal{E}'} g(X_i, X_j)
\end{equation}
where $g: \mathbb{R}^2 \to \mathbb{R}$ is the \textbf{interaction kernel.}
The interacting genes $\mathcal{E}'$ model the subset of the total interactions that significantly alter the response $T$. 
We refer to this full collection of objects as an \textbf{interaction network}, denoted~$\mathcal{N}$.

We may refer to a \textbf{mixed interaction network} when, for some subset of vertices $\mathcal{S} \subset \mathcal{V}$, and associated coefficient vector $\bm{\beta} \in \mathbb{R}^{|\mathcal{S}|}$, the response $T$ takes the form
\begin{equation}
\label{eq:networkMotivation.response.mixedNetwork}
T =  \sum_{(i,j) \in \mathcal{E}'} g(X_i, X_j) + \sum_{s \in \mathcal{S}} \beta_s X_s
\end{equation}
for a non-linear kernel $g$.
\end{defn}

Given an interaction network $\mathcal{N}$, our goal is to identify the pairs of interacting nodes $\mathcal{E}'$ from sample data drawn from $(\XX, T)$. 
Here, we seek a statistic that discriminates interacting pairs from non-interacting pairs, i.e., between the hypotheses $(i,j) \in \mathcal{E}'$ and $(i,j) \not\in \mathcal{E}'$. 
Given the expression of two genes $X_i$ and $X_j$ and the response $T$, the PID framework ideally allows us to decompose the mutual information between predictors and response into the four atoms of information in the bivariate PID lattice, as in (\ref{eq:PID_1}-\ref{eq:PID_3}).

To this inference task, we are interested in nominating edges via the different bivariate synergies induced by competing PIDs.
We will begin by employing these discrete functions upon discretized network simulations, to demonstrate their general behavior.
Discretizing continuous variables is typical in biological network science.
Although some of the early applications of information theoretic tools to biological network inference were justified by the natural extension to continuous variables, e.g. \cite{Butte2000RNAchemo}, it remains common practice to discretize many data types before applying information metrics\footnote{This practice should not necessarily be understood as a convenient concession. GRNs are commonly conceived as multi-valued logics, and discretization is consistent with such models.}.

\subsection{Interaction Network Simulations}
\label{subsection:network-simulations}

In order to evaluate competing PID synergies as tools for the inference of network interactions, we now present the results from simulations of interaction networks (Def.~\ref{defn:gene_interaction_network}).
For our interaction kernel $g$, we will use the following sigmoidal form:
\begin{equation}
    \label{eq:sigSw}
    g(x,y) = \frac{y}{1+e^{\alpha-x}}
\end{equation}
where $\alpha$ is a real parameter to be specified.
We will often refer to the first argument of this function as the `switch' node.
Sigmoidal functions are a typical choice, among others (including Hill functions), when modeling biomolecular activation in transcriptional networks \cite{Wang2007,Wang2014}. 

We simulate two very simple network topologies, as depicted in Fig~\ref{fig:network-topologies}.
In both, we have 50 nodes corresponding to predictors $(X_i)$, with edges signifying positive or negative correlation between them, depending on the sign attribute.
The network response $T$ is a function of a subset of these nodes, as in Eqs.~(\ref{eq:networkMotivation.response}, \ref{eq:networkMotivation.response.mixedNetwork}).
In Network A, 25 of the nodes are arranged in five 4-stars.
We take the edges on one of the stars as interactions.
That is, for one of the hubs $Y$ and its corresponding spokes $(X_i)_{i=1}^4$, we have $(X_i,Y) \in \mathcal{E}'$ and so the network has associated response:
\begin{equation}
\label{eq:networkResponse.networkA}
T = \sum_{i=1}^4 g(X_i,Y) = \sum_{i=1}^4 \frac{Y}{1+e^{\alpha-X_i}} .
\end{equation}
Network B is a mixed gene interaction network (Def.~\ref{defn:gene_interaction_network}).
In Network B, 20 of the nodes are arranged in two 10-stars, with hubs $Y_1$ and $Y_2$, and spokes $(X_i)_{i=1}^{10}$ and $(X_i)_{i=11}^{20}$, respectively.
We take $k$ of the spokes on $Y_1$ as interactions, i.e. $(X_i,Y_1) \in \mathcal{E}'$ for $1 \leq i \leq k$.
Moreover, we take the second hub $Y_2$ as a univariate signal to the response, weighted by real coefficient $\beta$.
Thus, we have the total network response:
\begin{equation}
\label{eq:networkResponse.networkB}
T = \frac{Y_2}{1+e^{\alpha-X_i}}  + \beta Y_2 .
\end{equation}
For Network B, we will often refer to the non-interacting edges $(X_i, Y_2)$ as false or fake interactions.
By adjusting $\beta$, we can ensure that a true interaction $(X_i, Y_1)$ and a false interaction $(X_j, Y_2)$ share the same mutual information with the response $T$.
Thus, these false interactions serve as an informationally-equivalent baseline.

\begin{figure}[h]
    \centering
    \includegraphics[scale=0.4, trim={1cm 3cm 1cm 2.5cm},clip]{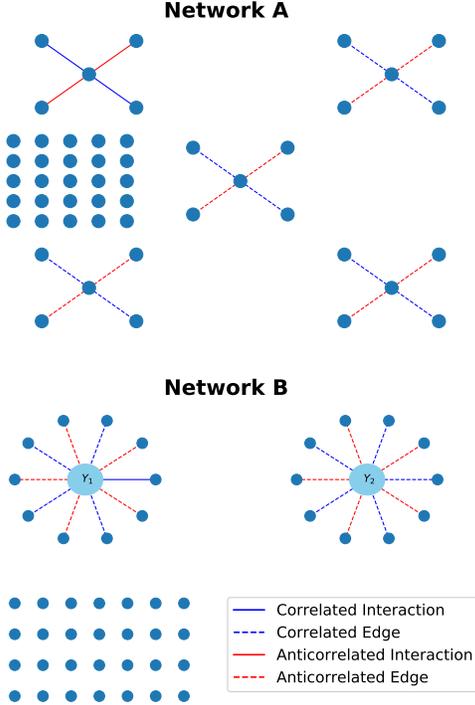}
    \caption{\textbf{Network Topologies A and B.} Represented are the topologies of Networks A and B, used in our simulations in Sec.~\ref{subsection:network-simulations}.}
    \label{fig:network-topologies}
\end{figure}

In Experiment I, we simulated 100 batches of $N=200$ realizations from Network A, setting the kernel parameter $\alpha=0$.
For each batch, we discretized each predictor $X_i$ and response $T$ independently into 3 equal-width bins.
We computed the full bivariate PID from the empirical distribution of $(X_i, X_j, T)$ for each pair of predictors $(X_i, X_j)$.
This allows us to assess the significance of each information atom for true interactions, relative to the distribution for null pairs.
We use four competing definitions of PID: the original $\Imin$ PID from \cite{Williams2010}, the $\Ipm$ PID from \cite{Finn2018main}, the $\Ibroja$ PID from \cite{Bertschinger2014}, and the $\Iccs$ PID from \cite{Ince2017}.

We present the ranked synergy scores of true interactions for each PID in Fig.~\ref{fig:experiment1}.
We see clearly that the $\Ipm$ synergy, alone, consistently ranks the synergy of interacting nodes above the 95th percentile of all synergy scores in a given batch.
Superficially, this might recommend $\Ipm$ synergy for this edge nomination task.
However, this result does not hold upon inspection.
In this experiment, MI itself $I(T;X_i, Y)$ likewise ranks true interactions above the 95th percentile.
Thus, $\Ipm$ synergy does not clearly improve upon the better-known information measure.
Pairwise MI, however, is well-known to be a problematic tool for network inference as it is unable to distinguish direct and indirect associations \cite{soranzo2007comparing,Chanda2020information}.
Thus, in Experiment II, we will instead simulate Network B in order to control for the role of mutual information.

\begin{figure}[h]
    \includegraphics[scale=0.6]{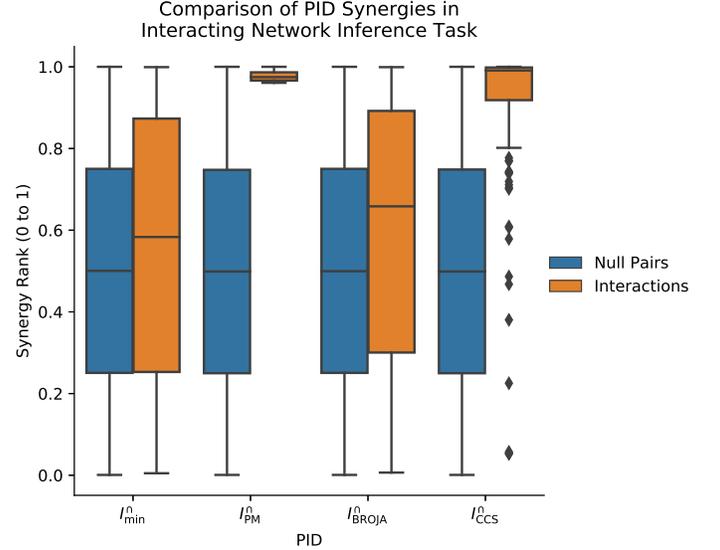}
     \caption{
     \textbf{Performance comparison of PID synergies for Experiment~I}
     We compare the performance of different PID synergies as discriminators of network interactions, as part of our first experiment with Network A in Fig~\ref{fig:network-topologies}. For each batch of 200 replicates, we computed the synergy atom of the bivariate PID (Eq.~\ref{eq:PID_1}) for each pair of genes in our network, and then converted these values into ranked scores from 0 to 1.  Shown here are the distributions of these scores, for interacting pairs and non-interacting pairs, the latter serving to approximate an empirical null distribution. As can be seen, only $\Spm$ consistently ranks true interactions in the top $5\%$ of gene pair synergies.}
    \label{fig:experiment1}
\end{figure}

As a final note on Experiment I, we comment upon the behavior of both synergy $S$ and the unique information $U_{X_i}$ PID atoms for the switch spoke $X_i$.
The $\Ipm$ PID assigns a higher absolute (non-ranked) amount of synergy to true interactions, while simultaneously assigning significant negative unique information $U_{X_i}$ to the switch node (Fig.~\ref{fig:supp.expI}).
Put differently, the inflated synergy $\Spm$ is due to the allowance of $\Upm{X}<0$ by the $\Ipm$ PID.
These two values balance each other out (\ref{eq:conservationBivariatePID}), as the four atoms are constrained by the mutual informations in (\ref{eq:PID_1}-\ref{eq:PID_3}).
The $\Imin$ and $\Ibroja$ PIDs obey the PID axiom of monotonicity (\ref{eq:WBaxiom.M}), which implies nonnegative unique information atoms in the bivariate decomposition.
The $\Ipm$ and $\Iccs$ PIDs do not obey the same restriction.

In Experiment II, we simulated Network B (Fig.~\ref{fig:network-topologies}), with a single interaction $(X_1, Y_1)$ and univariate signal $\beta Y_2$, varying $\beta$ in order to control the information provided by $Y_2$ (and thus all false interactions) toward the response $T$.
In Fig.~\ref{fig:experiment2}, we present the ranked synergy scores for the $\Imin$ and $\Ipm$ PID, as well as the ranked MI, for true and false interactions, i.e. $(X_1, Y_1)$ and $(X_j, Y_2)$, respectively.
We see that, as $\beta$ increases, both the MI and $\Ipm$ synergy of false interactions overtake true interactions.
The $\Imin$ synergy, however, consistently distinguishes true and false interactions.
The $\Imin$ synergy performs better when there is only one interaction to detect, unlike in the previous experiment.
We confirmed this with further Network B experiments in which we increased the number of interactions $k$.

\begin{figure}
    \centering
    \includegraphics[scale=0.4]{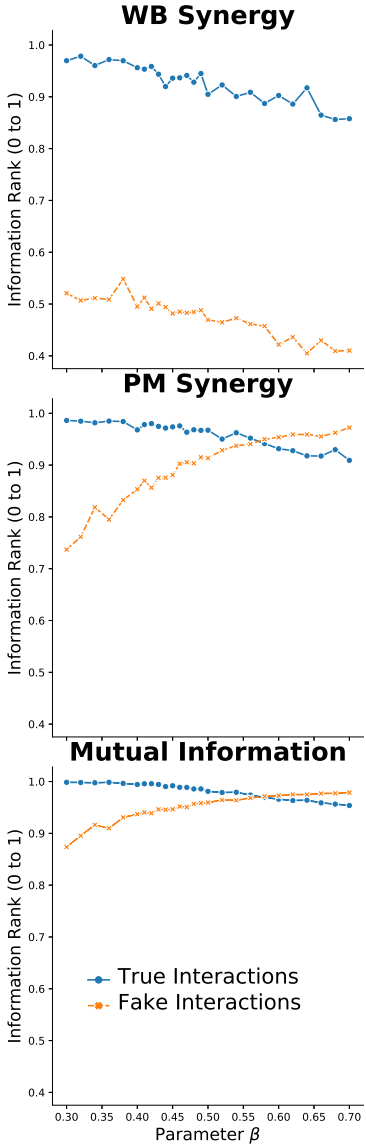}
\caption{
    \textbf{Relative synergy ($\Smin$ and $\Spm$) of true and false interactions in Exp.~II, compared to mutual information}
    From our simulations of Network B, we examine the ranked scores of synergy and mutual information as we vary the parameter $\beta$ from Eq.~(\ref{eq:networkResponse.networkB})). As $\beta$ increases, the univariate signal from $Y_2$ becomes more informative of the response than the interaction $g(X_1,Y_1)$. Hence, the MI of false pairs $I(T;X_j,Y_2)$ is greater than $I(T;X_1,Y_1)$ despite the irrelevance of $X_j$ to the response. $\Spm$ follows the same misleading trend, while $\Smin$ does not. Rather, $\Smin$ seems to firmly distinguish between true and false interactions even when the true interaction contributes less information about the response than $Y_2$.\\
    }
    \label{fig:experiment2}
\end{figure}

These results demonstrate that the $\Ipm$ synergy tracks closely to the mutual information, and has issues with specificity.
Since a major challenge in many network inference tasks lies in the difficulty of distinguishing direct and indirect associations, this situation suggests a major drawback to applying the $\Ipm$ PID to network inference.
This non-specificity can be further seen in Suppl. Fig.~\ref{fig:supp.expII}, in which all bivariate atoms of the $\Ipm$ PID seem to track with mutual information, for true and false interactions, while the $\Imin$ PID demonstrates qualitatively different behavior toward the two.
Although the $\Imin$ synergy may not have demonstrated sensitivity in our previous simulations, it specifically distinguishes bivariate interactions from univariate signals.
As we will see, this is a structural difference between the two PIDs.
As mutual information $I(T;X,Y)$ increases, we again see an increase in $\Ipm$ synergy (and redundancy) at the expense of negative unique information on the less informative predictor $X$, now for both actual interactions and univariate signals (Fig.~\ref{fig:supp.expII}).

We have seen that the $\Imin$ PID distinguishes between true and false interactions.
What has been demonstrated is the special case of a more general principle, which applies to all nonnegative PIDs (e.g. discrete $\Imin$ and $\Ibroja$, but not $\Ipm$).
Such PIDs will necessarily respect the conditional independence of predictors, which is a crucial property for the analysis of high-dimensional data.
We present this elementary result.

\begin{proposition}
\label{prop:PID.NN_implies_zero_synergy}
Let $X, Y, T$ be jointly discrete or continuous random variables.
Suppose further that $X$ is conditionally independent of $T$ given $Y$, i.e. $X \perp T | Y$.
Let $\mathfrak{d}$ be a bivariate PID of $I(T; X,Y)$, induced by a redundancy function $\Ired$ that induces a fully non-negative PID.
Then we have that:
\begin{align*}
    I(T; X, Y) &= I(T; Y) \\
    S(T; X, Y) &= U_X(T; X, Y) = 0
\end{align*}
\end{proposition}

More generally, if an $\Ired$ induces a nonnegative $\Pi$ function as in \cite{Williams2010}, i.e. the function that assigns all PID atoms in an information lattice, then we will have analogous results for higher-order conditional independencies.
Here, we demonstrate the idea in the bivariate case.

\begin{proof}
By the chain rule for MI, we have that
\begin{equation}
    I(T;X,Y) = I(T;Y) + I(T;X|Y).
\end{equation}
Since $X$ and $T$ are conditionally independent, we have that $I(T;X|Y) = 0$, and thus
\begin{equation}
    I(T;X,Y) = I(T;Y).
\end{equation}
Usings Eqs.~(\ref{eq:PID_1})-(\ref{eq:PID_3}), we have
\begin{equation}
    R + U_X + U_Y + S = R + U_Y.
\end{equation}
Since all atoms are non-negative, it follows that $U_X = S = 0.$
\end{proof}

In Experiment II, the false interactions $(X_j, Y_2)$ form such a conditional independence $X_j \perp T | Y_2$, but the $\Ipm$ does not satisfy the non-negativity condition.
Thus, we see in Fig.~\ref{fig:supp.expII} that neither $U_X$ nor $S$ are zero for false interactions.

\subsection{Analytic and Informational Balance Between Predictors}
\label{subsection:experimentC}

Before we turn to extending PID to continuous interactions in the next section, we may bridge this gap by examining a correspondence between the analytic properties of the sigmoidal kernel (\ref{eq:sigSw}) and the discretized PIDs of our network simulations. Although we may have discretized our simulation data, they are the realization of a continuous interaction network model, with a response smoothly determined by the random predictors.
In this idealized setting, we may ask how differing conceptions of PID handle this smooth determination, particularly with respect to kernel derivatives.
Crucially, in the analytic `information' that determines our response --- i.e. the coefficients of the power series expansion of (\ref{eq:networkMotivation.response}) , which fully determine $T$  --- we may cleanly distinguish between univariate derivatives and cross-terms.
While there is no direct correspondence between the atoms of PID and the kernel interaction derivatives, nevertheless, as we will show below, there is a delicate relationship between the two.

Our choice of a sigmoidal switch interaction kernel corresponds to an approximation of Boolean or multivalued logic circuits within the 'omic regulatory system of a biological organism.
The gene $X$ in Eq.~(\ref{eq:sigSw}) is the `switch' gene, with higher values of $X$ turning the interaction on, and lower values turning it off.
We allow the parameter $\alpha$ to recenter this transition, which is equivalent to recentering $X$.
In Experiment~I, we chose $\alpha=0$, which allows both the `on' and `off' regime to be observed with high probability.
However, by choosing another $\alpha$, the interaction can be made to default to one state or another, thereby shifting the balance of information between the predictors.
For low $\alpha$, the effect of the $Y$ predictor is somewhat suppressed, and small changes in the switch $X$ will affect the response $T$.
For high $\alpha$, the sigmoidal function is `on' (with high probability) despite perturbations in $X$, and so we expect $I(T;Y) \gg I(T;X)$.
We would expect this difference in relative information to be qualitatively apparent in any PID.

Before adjusting $\alpha$ in simulations, we may approach this $\alpha$-dependent information balance heuristically.
A second-order Taylor expansion of our kernel from Eq.~(\ref{eq:sigSw}) about 
$(0,0)$ supplies
\begin{align}
\label{eq:sigSw.2ndOrderExpansion}
    g(x,y) \approx \underbrace{\frac{1}{1 + e^{\alpha}}}_{\partial_{y} g} y + \underbrace{\frac{e^{\alpha}}{(1 + e^{\alpha})^2}}_{\partial_{x,y} g} xy .
\end{align}
We may imagine the coefficient $\partial_y g$ as analogous to the unique information PID atom $U_Y(T;X,Y)$.
The coefficient $\partial_{x,y}$ is more ambiguous, as it may be argued to contribute to both redundant and synergistic information.
We will return to this matter shortly.
The term $\partial_y g$ dominates $\partial_{x,y}g$ for lower $\alpha$, but the terms converge as $\alpha$ increases.
That is,
\begin{align}
    (\partial_{y} g, \partial_{x,y} g) &\sim c(\alpha) \left( \frac{1}{2}, \frac{1}{2} \right) \text{ as } \alpha \to \infty,\\
    (\partial_{y} g, \partial_{x,y} g) &\sim c(\alpha) \left( 1, 0 \right) \text{ as } \alpha \to -\infty,\\
    \vspace{6pt} \nonumber \text{ where }  c(\alpha) &= |\partial_{y} g| + |\partial_{x,y} g|.
\end{align}
Our response for Network A in Eq.~(\ref{eq:networkResponse.networkA}) is the sum of four such interactions:
$$
T = f(X_1,X_2,X_3,X_4, Y) = \sum_i g(X_i,Y).
$$
Expanding this gives us
\begin{align}
\label{eq:networkResponse.networkA.2ndOrderExpansion}
    f(x_1,x_2,x_3,x_4,y) \approx \underbrace{\frac{4}{1 + e^{\alpha}}}_{\partial_{y}} y
    + \sum_{i=1}^4 \underbrace{\frac{e^{\alpha}}{(1 + e^{\alpha})^2}}_{\partial_{x_i,y} g} x_i y
\end{align}
and thus
\begin{align}
    \label{eq:networkResponse.networkA.2ndOrder.ratio.highAlpha}
    (\partial_{y}f, \partial_{x_i,y}f) &\sim c(\alpha) \left(\frac{1}{2} , \frac{1}{8} \right) \text{ as } \alpha \to \infty,\\
    \label{eq:networkResponse.networkA.2ndOrder.ratio.lowAlpha}
    (\partial_{y}f, \partial_{x_i,y}f) &\sim c(\alpha) \left(1, 0 \right) \text{ as } \alpha \to -\infty,\\
    \nonumber \text{ where }  c(\alpha) &= |\partial_{y} g| + \sum_{i=1}^4 |\partial_{x_i,y} g|.
\end{align}
More sigmoidal interactions on the same hub will divide the `switch' information among the spokes.
Nonetheless, we have the same qualitative difference between high- and low-$\alpha$ regimes as with a single interaction kernel.

We might expect, then, that in simulations of Network A, a PID of the information of an interaction pair $I(T;X_i, Y)$ would locate most of the MI in the unique information atom $U_{Y}$ for low values of $\alpha$.
By contrast, for higher values of $\alpha$, information
ought to become more entangled with the switch nodes $X_i$.
In Experiment III, we simulated Network A again while varying the parameter $\alpha$, in order to explore the changing balance of PID information between nodes.

We see in Fig.~\ref{fig:experiment3} that $\Imin$ \textit{and not} $\Ipm$ locates most of the information in $\Umin{Y}$ for low values of $\alpha$.
We directly compare the normalized values of $\Umin{Y}$ and $\Upm{Y}$ to the relative value of $\partial_y f$ as in Eqs.\@ (\ref{eq:networkResponse.networkA.2ndOrder.ratio.highAlpha}) and (\ref{eq:networkResponse.networkA.2ndOrder.ratio.lowAlpha}), and see that $\Umin{Y}$ nearly tracks this expression, unlike $\Upm{Y}$.
This result favors the $\Imin$ framework for demonstrating intuitive sensitivity to the analytic predominance of a single predictor.

\begin{figure*}
    \centering
    \includegraphics[scale=0.4]{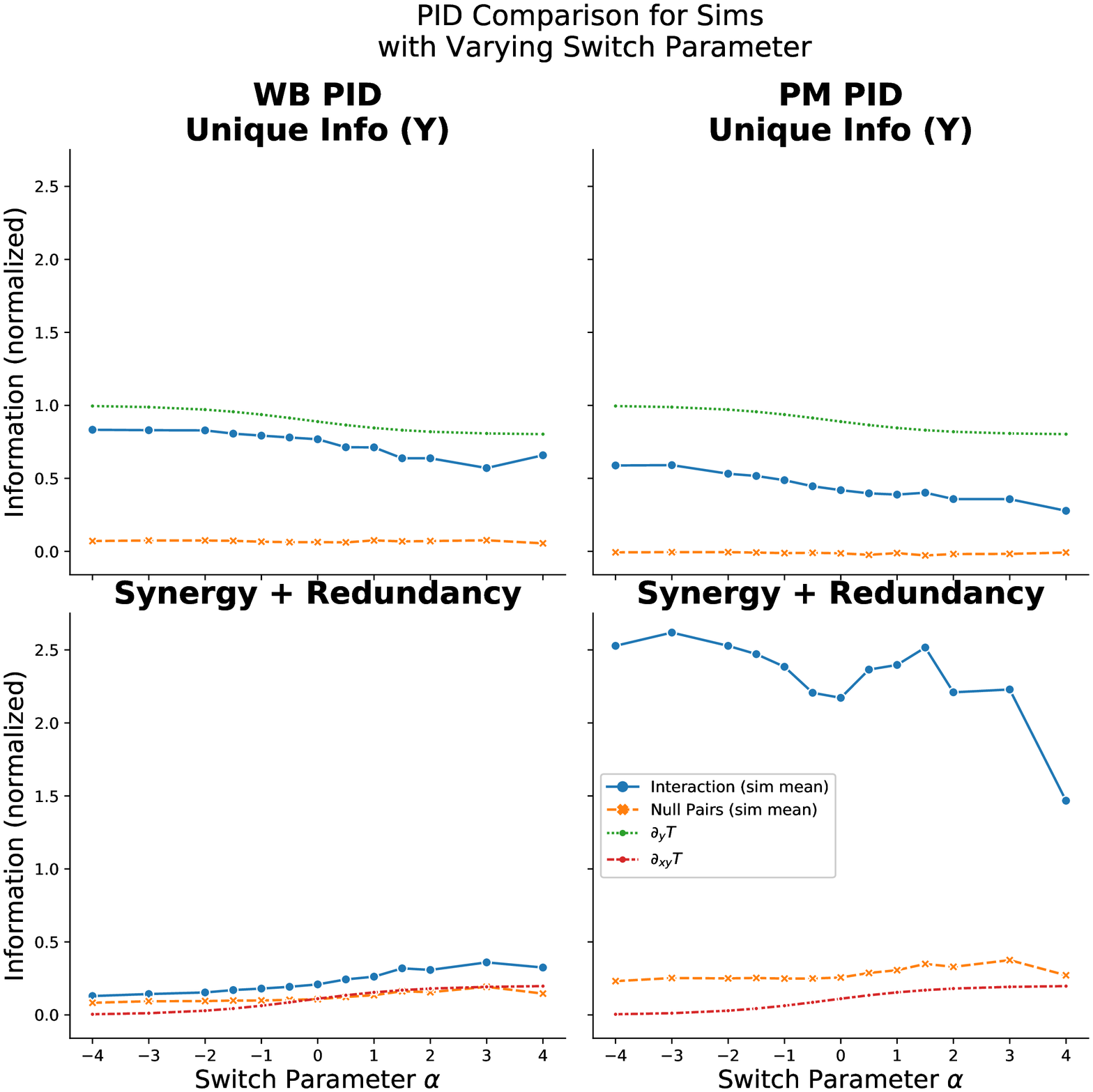}
     \caption{
     \textbf{Relationship between bivariate PID atoms and sigmoidal kernel derivatives in Experiment~III}
     We examine both the unique info $U_X$ of the hub gene $Y$ and the sum of synergy and redundancy $S+R$ for our interaction pairs $(X_i, Y)$, in light of the expansion of the response in Eq.~(\ref{eq:networkResponse.networkA.2ndOrderExpansion}). We see that, as we vary $\alpha$, $\Umin{Y}$ tracks the normalized term $ \partial_y f /c(\alpha),$ where $ c(\alpha) = |\partial_y f| + |\partial_{xy} f|$, while $\Smin+\Rmin$ tracks $\partial_{xy} f/c(\alpha)$. The former term corresponds  to the linearized sensitivity of the response $T$ to $Y$ alone near $X=Y=0$, while the latter tracks the sensitivity to the product $XY$.
     }
    \label{fig:experiment3}
\end{figure*}

In Fig.~\ref{fig:experiment3}, we also see that, for the $\Imin$ PID, the sum of redundant and synergistic information $\Rmin + \Smin$ follows a similar pattern as $\partial_{xy} f$, increasing monotonically in $\alpha$, while $\Rpm+\Spm$ remains constant.
Both $\Rmin$ and $\Smin$ increase individually with $\alpha$, with $\Rmin > \Smin$ (Fig.~\ref{fig:supp.expIII}).
As mentioned before, it not immediately clear whether the joint sensitivity of $T$ to both nodes in an interaction, captured by the coefficient $\partial_{xy} f$ in Eqs.~\ref{eq:sigSw.2ndOrderExpansion}-\ref{eq:networkResponse.networkA.2ndOrderExpansion}, ought to be accounted for as synergistic or as redundant information.
In expectation, that term controls the dependency of $T$ on the covariance $\mathbb{E} X Y$. 
If linear correlation between predictors is understood as information shared between them, then this would suggest that $\partial_{x,y} f$ ought to correspond to PID redundancy.
On the other hand, as a measure of mutual sensitivity of the kernel $g$, this term also captures the two-dimensional effect of small perturbations in both variables.
If perturbations in $X$ likewise increase sensitivity of $g$ to perturbations in $Y$, an argument could be made for synergy.

Despite this ambiguity in the PID-based interpretation of $\partial_{x,y} g$, the behavior demonstrated by the $\Imin$ PID better aligns with the analytic sensitivities of our target variable.
By contrast, the $\Ipm$ PID divides MI into stable proportions independent of any changes to relative sensitivity of the internal kernel upon the predictors in question.
In other words, the $\Ipm$ PID lacks specificity in its assignment of predictor-target mutual information to the redundant and synergistic atoms, at the expense of the unique information atoms.
This concords with our observations from Experiment II in Sec.~\ref{subsection:interaction_network}, where the $\Imin$ but not the $\Ipm$ PID synergy could distinguish true and false interactions, conflating high pair-wise mutual information $I(T;X,Y)$ with truly synergistic information.
As we will see in the next section, this behavior is structural to the competing definitions of the $\Imin$ and $\Ipm$ PIDs.
By extending these definitions to continuous variables, we may demonstrate that $\Imin$ but not $\Ipm$ demonstrates a respect for conditionally independent predictors.
Moreover, we may characterize the allotment of unique versus redundant information using only the relative sensitivities of a target interaction outcomes (the pointwise functions $\partial_x g, \partial_y g$), as well as the marginal information content of the predictors themselves.

\section{Unique Information of the $\Imin$ and $\Ipm$ PIDs for Continuous Interactions}
\label{section:GenericResults}

As we saw in the previous section, the $\Ipm$ PID is non-specific in evaluating network interactions, while the $\Imin$ PID is more specific but insensitive.
In Experiment II, the $\Ipm$ PID could not distinguish between interacting nodes jointly affecting the response variable and a pair in which only one node contributed a univariate signal.
In other words, despite the conditional independence $T \perp X | Y$, $\Spm$ and $\Rpm$ were given a greater share of the total information $I(T;X,Y) = I(T;Y)$.
In Experiment III, we demonstrated a heuristic approach to analytically separating the unique and joint (redundant + synergistic) information in sigmoidal interactions, and saw in simulation that the $\Imin$ PID agreed with this separation, while the $\Ipm$ PID did not.
In other words, the $\Imin$ PID agreed with the analytic balance of information between predictors, and the $\Ipm$ PID ignored changes to this balance.
In this section, we extend both PIDs to continuous interactions in order to demonstrate the general principles underlying this difference.
These principles apply to a more general class of interaction kernels, and extend beyond the simple networks used in our simulations.

\subsection{Noise-free Bivariate Interactions}
\label{subsection:NFBI}

Due to Eqs.~(\ref{eq:PID_1}-\ref{eq:PID_3}), the difference between two finite PIDs in the values assigned to each atom is conserved, according to Eq.~(\ref{eq:conservationBivariatePID}).
Thus, we use a noise-free interaction model that allows us to directly understand how the redundant and unique information atoms behave toward continuous interaction kernels.
We will examine a trivariate system of random variables, two predictors $X$ and $Y$, and target $T$, such that $T$ is deterministic interaction of $X$ and $Y$ for some kernel $g$:
\begin{equation}
\label{eq:motivation.noisfreeModel}
    T = g(X,Y).
\end{equation}
We will refer to such a system as a noise-free bivariate interaction.

\begin{defn}[Noise-free Bivariate Interaction]
\label{defn:bivarIntxn.noisefree}
Let $X,Y, T$ be real-valued random variables and $g$ be a continuous function such that Eq.~(\ref{eq:motivation.noisfreeModel}) holds for every realization, i.e.
\begin{equation}
    T(\omega) = g(X,Y)(\omega) \text{ for every } \omega \in \Omega.
\end{equation}
We call the triplet $(X,Y,T)$ a \textbf{noise-free bivariate interaction (NFBI)}, and denote it $X,Y \to T$.
If $(X,Y) \sim N(\bm{\mu}, \Sigma)$, we refer to the interaction as \textbf{Gaussian}, denoted $X,Y \to_\Sigma T$. If $\bm{\mu} = \bm{0}$ and $\Sigma = (1, \rho; \rho, 1)$, we refer to the interaction as \textbf{normalized Gaussian}, denoted $X,Y \to_\rho T$.
\end{defn}
A NFBI will necessarily have a degenerate distribution on the triplet $(X,Y,T)$.
Even when $\mu_{X,Y} \ll \lambda^{(2)}$ and $(X,Y)$ admit a density, the triplet will have a point-mass distribution:
\begin{equation}
    \label{eq:motivation.noisefreeStatModel.pxyz}
    p_{X,Y,T}(x,y,t) = p_{X,Y}(x,y) \delta(t - g(x,y)) .
\end{equation}
Thus, $I(T;X,Y) = \infty$ for any NFBI (Cor.~\ref{corollary:infiniteMI.sumOfDiracs}).
However, when $I(T;X)$ and $I(T;Y)$ are finite due to the nonsingular smoothness of $g$, we suspect any (non-indeterminate) PID will assign finite information to every atom except synergy.

We will always assume that the predictors $(X,Y)$ of a NFBI admit a smooth density $p_{X,Y}$ that is absolutely continuous with respect to Lebesgue measure.

\begin{condition}[Well-behaved Gaussian NFBI]
\label{condition}
Let $(X,Y) \rightarrow T$ be a NFBI with associated kernel $g$.
We say that $(X,Y) \rightarrow T$ is \textbf{well-behaved} if there exists an open set $U \subset \mathcal{A}_{X,Y}$ of full measure $\mu_{X,Y}(U)=1$ such that $g$ satisfies the following conditions:
\begin{itemize}
\item[i.] The function $g$ is continuously differentiable on $U$, i.e. $g \in C^{(1)}(U)$;
\item[ii.]  The partials
$|\partial_y g(x,y)|>0$,  
$|\partial_x g(x,y)|>0$ for $(x,y)\in U$.
\end{itemize}
\end{condition}
\noindent This condition 
is sufficiently flexible to encompass, for instance, our kernel in Eq.~(\ref{eq:sigSw}).
For our purpose, we need this condition to guarantee the following change of variable formula.

\begin{proposition}
\label{proposition:GenericKernel.CoV}
Let $(X,Y) \rightarrow T$ be a well-behaved NFBI as per Condition~\ref{condition}.
Then the densities on the induced probability spaces $(\mathcal{A}_{X,T},\mu_{X,T})$ and $(\mathcal{A}_{Y,T},\mu_{Y,T})$ are defined a.e. (where we write $\partial_y g$ for $\partial_y g(x,\tilde y)$ and similarly for $\partial_x g$)
\begin{align}
    \label{eq:GenericKernel.CoV.pxt}
    p_{X,T}(x,t) &= \frac{1}{|\partial_y g|} p_{X,Y}(x, \tilde{y} ) \text{ a.e. } [\mu_{X,T}],\\
    \label{eq:GenericKernel.CoV.pyt}
    p_{Y,T}(y,t) &= \frac{1}{|\partial_x g|} p_{X,Y}(\tilde{x} , y) \text{ a.e. } [\mu_{Y,T}],
    \intertext{where}
    \tilde{y}(x,t) &= (g(x,\cdot))^{-1}(t),\\
    \tilde{x}(y,t) &= (g(\cdot,y))^{-1}(t).
\end{align}
Equivalently, we may regard these densities as random variables that are almost surely equal (where we write $\partial_y g$ for $\partial_y g(X,Y)$ and similarly for $\partial_x g$):
\begin{align}
    p_{X,T}(X,T) &\oeq \frac{1}{|\partial_y g|} p_{X,Y}(X, Y ) \\
    p_{Y,T}(Y,T) &\oeq \frac{1}{|\partial_x g|} p_{X,Y}(X , Y)
\end{align}
\end{proposition}
\begin{proof}
Let $U\subset \mathcal{A}_{X,Y}$ be the full measure open set as in the statement of Condition~\ref{condition}.
Consider the transformation $\Phi:U \to \mathbb{R}^2$ given by $\Phi: (x,y) \mapsto (x,g(x,y))$.
Since $g \in C^{1}(U)$, the Jacobi matrix $J_\Phi$ is defined, and we have that $|\det J_{\Phi}| = |\partial_y g| > 0$ on $U$.
We arrive at our density $p_{X,T}$ by a standard change of variables (see, for instance, Section~4.3 of Casella and Berger \cite{CasellaBerger}).
\end{proof}

\subsection{Continuous PIDs of a Well-behaved Noise-free Bivariate Interaction}
\label{subsection:PIDofNFBI}

Now that we have defined NFBIs and provided their joint marginal representations, we may state our main result, Theorem~\ref{theorem:GenericKernel.UniqueInfo}.
This theorem describes the bivariate PID of NFBIs by characterizing the unique information atoms in terms of the first-order derivatives of the interaction kernel $g$, as well as the marginal distributions $p_X(X)$ and $p_Y(Y)$.
As we will discuss, this characterization will allow us to understand the non-specificity of the $\Ipm$ PID and the insensitivity of the $\Imin$ PID as tools of interaction nomination.

\begin{theorem}[Unique Information for Continuous $\Imin$ and $\Ipm$ PIDs]
\label{theorem:GenericKernel.UniqueInfo}
Let $(X,Y) \rightarrow T$ be a well-behaved noise-free bivariate interaction with associated kernel $g$, as per Condition~\ref{condition}.
Then the unique informations for the decompositions $\dmin$ and $\dpm$ are given by

\begin{align}
   \Umin{X}(T;X,Y) &=
    \mathbb{E} \left[ \mathds{1}_{A}
    \left( \log \frac{p_Y(Y)}{p_X(X)} - \log \frac{|\partial_y g|}{|\partial_x g|}
    \right) \right]
    \label{eq:GenericKernel.Thm.Umin}\\
    \Upm{X}(T;X,Y) &=
    \underbrace{\mathbb{E} \left[
    \mathds{1}_{B} \log \frac{p_Y(Y)}{p_X(X)}
    \right]}_{(\Upm{X})^+} -
    \underbrace{\mathbb{E} \left[
    \mathds{1}_{C} \log \frac{|\partial_y g|}{|\partial_x g|}
     \right]}_{(\Upm{X})^-}
    \label{eq:GenericKernel.Thm.Upm}
\end{align}

where $\mathds{1}$ is the indicator function for the events:
\begin{align*}
    A &= \left\{I_X(T) \geq I_Y(T)  \right\}\\
    B &= \left\{p_X(X) \leq p_Y(Y) \right\}\\
    C
    &= \left\{|\partial_x g|(X,Y) \leq |\partial_y g|(X,Y) \right\} 
\end{align*}

In particular, if $(X,Y)_\rho \rightarrow T$ is a normalized Gaussian interaction, then we have that
\begin{align}
\Umin{X}(T;X,Y)
&=
\mathbb{E} \left[ \mathds{1}_{A} \left( \frac{\log(e)}{2} (X^2 - Y^2) - \log \frac{|\partial_y g| }{|\partial_x g| } \right) \right],
\label{eq:GenericKernel.Thm.Umin.Gaussian}\\
 \Upm{X}(T;X,Y)
        &= \frac{\log(e)}{2} \sqrt{1-\rho^2}
        -
        \mathbb{E} \left[
        \mathds{1}_{C} \log \frac{|\partial_y g|}{|\partial_x g|}
         \right] .
         \label{eq:GenericKernel.Thm.Upm.Gaussian}
\end{align}
\end{theorem}

From this theorem, we immediately see how $\Ipm$ can assign negative unique information. If the ratio of partial derivatives $\frac{|\partial_x g|}{|\partial_y g|}$ is arbitrarily small on most of the probability space, as we would expect with a kernel with very assymetric dependency on $X$ and $Y$, then $\Upm{X}$ will be negative.
Moreover, the unique specificity $\Upmm{X}$ will be infinite if $\partial_x g = 0$ on a non-trivial set for which $\partial_y g \neq 0$.
In this case, $\Upm{X} = - \infty$.
We will return to this topic in Sec.~\ref{subsection:GenericResults.asymptoticCI}.

\begin{proof}[Proof of Theorem~\ref{theorem:GenericKernel.UniqueInfo}]
Let $g$ and $U$ be given as in the statement of Condition~\ref{condition}.
On the almost sure event $E = \pi_{X,Y}^{-1}(U)$, note that via Proposition~\ref{proposition:GenericKernel.CoV}, the densities $p_{X,Y}, p_{X,T},$ and $p_{Y,T}$ are well-defined.
Moreover, since we can choose $U$ to ensure $p_T(t)>0$ over $(x,y)\in U$, the conditional densities $p_{X|T}$ and $p_{X|T}$ are likewise defined.

We begin with (\ref{eq:GenericKernel.Thm.Umin}).
We recall from Def.~\ref{defn:redundant-info.Imin}
that $\Imin(T;X,Y) = \mathbb{E} \min(I_X(T),I_Y(T))$, so
\begin{align*}
\Umin{X} &= I(T;X) - \Imin(T;X,Y)\\
&= \mathbb{E} I_X(t) - \mathbb{E} \min(I_X(T),I_Y(T)) \\
&= \mathbb{E} \left[ \mathds{1}_{A} (I_X(T) - I_Y(T)) \right]
\end{align*}
where $A$ is the event in which $I_X(t) - I_Y(t)>0$.
If $p_T(t) = 0$, $I_X(t) = I_Y(t) = 0$, so we may assume $p_T(t) > 0$.
Then observe that
\begin{align}
    &\left( I_X(T) - I_Y(T) \right)\notag\\
    &= \mathbb{E}\left( \log \frac{p_{X,T}(X,T)}{p_X(X) p_T(T)}\,\bigg|\,T\right) \notag\\
    \nonumber
    &\hspace{10mm}- \mathbb{E}\left( \log \frac{p_{Y,T}(Y,T)}{p_Y(Y) p_T(T)}\,\bigg|\,T\right) \\
    \nonumber
    &=  \mathbb{E}\left( \log \frac{p_{X,T}(X,T)}{p_X(X)}  
    - \log \frac{p_{Y,T}(Y,T)}{p_Y(Y)} \,\bigg|\,T \right) \\
    \nonumber
    &= \mathbb{E}\left( \log \frac{p_Y(Y)}{p_X(X)} 
    - \log \frac{p_{Y,T}(Y,T)}{p_{X,T}(X,T)} \,\bigg|\,T \right)
\end{align}

By Proposition~\ref{proposition:GenericKernel.CoV}, we have that
\begin{align*}
    \log \frac{p_Y(Y)}{p_X(X)} 
    - \log \frac{p_{Y,T}(Y,T)}{p_{X,T}(X,T)}
    &\oeq
    \log \frac{p_Y(Y)}{p_X(X)} 
    - \log \frac{|\partial_y g|}{|\partial_x g|}
\end{align*}
Thus, considered as random functions of $T$,
\begin{align}
I_X(T) - I_Y(T) \oeq \mathbb{E} \left( \log \frac{p_Y(Y)}{p_X(X)} 
    - \log \frac{|\partial_y g|}{|\partial_x g|} \,\bigg|\,T\right).
\end{align}
Now, $\mathds{1}_A$ is a function of $T$ so that
\begin{align*}
    &\mathbb{E} \left[ \mathds{1}_{A}
    \mathbb{E}\left( \log \frac{p_Y(Y)}{p_X(X)} - \log \frac{|\partial_y g|}{|\partial_x g|}\,\big|\, T
    \right) \right]\\
    &=\mathbb{E} \left[ 
    \mathbb{E}\left( \mathds{1}_{A}\left(\log \frac{p_Y(Y)}{p_X(X)} - \log \frac{|\partial_y g|}{|\partial_x g|}\,\big|\, T
    \right)\right) \right]\\
    &=\mathbb{E} \left[ 
     \mathds{1}_{A}\log \frac{p_Y(Y)}{p_X(X)} - \log \frac{|\partial_y g|}{|\partial_x g|} \right]
\end{align*}
and Eq.~\ref{eq:GenericKernel.Thm.Umin} follows. 

The expression for unique specificity $\Upmp{X}$ in Eq.~\ref{eq:GenericKernel.Thm.Upm} follows easily from Def.~\ref{defn:redundant-info.Ipm} and Eq.~\ref{eq:redundant-info.Ipm:eqn:sublattice.Ux}.
To find the expression for unique ambiguity in Eq.~\ref{eq:GenericKernel.Thm.Upm}, we have that
\begin{align*}
\Upmm{X} &= \Ipmm(T;X) - \Ipmm(T;X,Y)\\
&= \mathbb{E} \log \frac{1}{p_X(X)} - \mathbb{E} \min(\frac{1}{p_{X|T}(X|T)},\frac{1}{p_{Y|T}(Y|T)}) \\
&= \mathbb{E} \left[ \mathds{1}_{C}  \log \frac{p_{Y|T}(Y|T)}{p_{X|T}(X|T)} \right] 
\end{align*}
where $C \subset E$ is the event
$$
C = \{ p_{X|T}(X|T) > p_{Y|T}(Y|T) \} \cap E.
$$
Note that $C$ may be empty or probability zero, in which case $\Upmm{X} = 0$.
Prop.~\ref{proposition:GenericKernel.CoV} gives us the almost sure equality of the random variables
\begin{align}
    \log \frac{p_{X|T}(X|T)}{p_{Y|T}(Y|T)} \oeq \log \frac{|\partial_x g (X,Y)|}{|\partial_y g (X,Y)|}.
\end{align}
Thus, our expression in Eq.~(\ref{eq:GenericKernel.Thm.Upm}) follows, as does our definition of the event $C$ stated in the theorem.

For the specific case where $(X,Y)_\rho \to T$, then we have that both marginals are the standard normal Gaussian density $p_X(z) = p_Y(z) = \frac{1}{\sqrt{2\pi}} e^{-z^2/2}$. Thus,
$$
\log \frac{p_Y(y)}{p_X(x)} = \frac{\log(e)}{2}(x^2 - y^2)
$$
and Eq.~(\ref{eq:GenericKernel.Thm.Umin.Gaussian}) follows.
The expression for $\Upmp{X}$ was demonstrated in Lemma~\ref{lemma:linearIntxn.PM.lattice.specificity} of Supp.~\ref{subsection:appendix.computeLinearPIDs}.
\end{proof}

\subsection{Asymptotic Conditional Independence of Predictors in the Continuous $\Imin$ and $\Ipm$ PIDs}
\label{subsection:GenericResults.asymptoticCI}

Throughout Section~\ref{section:SynergyNetworks}, the question was raised regarding the appropriate allocation of MI to PID information atoms, taking into consideration the analytic sensitivity of the response $T$ on the predictors.
In Experiment II (Sec.~\ref{subsection:network-simulations}), we went to the extreme case in which only one of the two predictors under examination, $Y_2$, had any impact on the response $T$, while the other (any spoke $X_j$ on $Y_2$) was conditionally independent of the response, given $Y_2$.
Intuitively, we argued that $\Imin$, by assigning the bulk of the total response-predictor mutual information to the unique information PID atom $U_Y$, was respecting conditional independence, while $\Ipm$ was not.
In Sec.~(\ref{subsection:experimentC}), we took a more careful look at the balance of information between variables.
By adjusting the kernel parameter $\alpha$, we continuously varied the relative dependency of the response on each variable (Fig.~\ref{fig:experiment3}).
We saw that the $\Imin$ PID respected these  changes in relative dependency, continuously adjusting its allocation with $\alpha$.
By contrast, the $\Ipm$ PID assigned consistent proportions of the total information to all PID atoms, irrespective of the changing balance of information between variables.
Taken together, our experiments highlight the difference in how each PID treats conditionally independent (and almost independent) predictors.

To explain this difference, we now develop the asymptotic consequences of our result from the previous subsection.
For a response $T$ to be conditionally independent of $X$, we must have that $|\partial_x g| \oeq 0$, which would violate Condition~\ref{condition}.
Thus, our results in this subsection serve to extend Theorem~\ref{theorem:GenericKernel.UniqueInfo} to the limit points of admissible kernels.
This will include univariate signals as in Sec.~\ref{subsection:network-simulations}.

\begin{corollary}
\label{corollary:GenericKernel.limits}
Let $\{ X,Y \to_{\rho} T_{(n)} \}_{n \in \mathbb{N}}$ be a sequence of noiseless Gaussian interactions, associated to the sequence of kernels $\{ g_{n} \}$.
Let us further assume that each $g_{n}$ satisfies Condition~\ref{condition}, on the same common, full measure open set $U \subset \mathcal{A}_{X,Y}$.
Assume also that $\Umin{X},\Umin{Y}, \Upm{X},\Upm{Y}$ are all finite for $X,Y \to_{\rho} T_{(1)}$.

Then the following monotonic limit of random variables
\begin{equation}
\label{eq:GenericKernel.limits.assumed_limit}
\frac{|\partial_x g_{(n)}|}{|\partial_y g_{(n)}|}(X,Y) \downarrow 0 \text{ as } n \to \infty \text{ almost surely}
\end{equation}
implies the following information limits:
\begin{align}
\label{eq:GenericKernel.limits.UminX}
\Umin{X}(T_{(n)}; X,Y) \to 0,\\
\label{eq:GenericKernel.limits.UpmX}
\Upm{X}(T_{(n)}; X,Y) \to - \infty, \\
\label{eq:GenericKernel.limits.UpmmX}
\Upmm{X}(T_{(n)}; X,Y) \to \infty,\\
\label{eq:GenericKernel.limits.UminY}
\Umin{Y}(T_{(n)}; X,Y) \to \infty,\\
\label{eq:GenericKernel.limits.UpmY}
\Upm{Y}(T_{(n)}; X,Y) \to c(\rho) ,\\
\label{eq:GenericKernel.limits.UpmmY}
\Upmm{Y}(T_{(n)}; X,Y) \to 0,
\intertext{where}
\nonumber c(\rho) = \frac{\log(e)}{2} \sqrt{1-\rho^2}.
\end{align}
\end{corollary}

\begin{proof}
In this proof, we will use $A_n$, $B_n$, and $C_n$ to denote the events defined in the statement of Theorem~\ref{theorem:GenericKernel.UniqueInfo}, corresponding to each kernel $g_n$.

We first demonstrate Eq.~(\ref{eq:GenericKernel.limits.UminX}), using expression for $\Umin{X}$ from Theorem~\ref{theorem:GenericKernel.UniqueInfo}.
Consider the events
\begin{align*}
E &= \left\{ \omega \text{ s.t. } \frac{\partial_x g_n}{\partial_y g_n}(X,Y)(\omega) \downarrow 0 \right\} \\
\tilde{E} &= E \cap \pi_{X,Y}^{-1} U
\end{align*}
and note that $\tilde E$ has full measure.
We thus restrict ourselves to $\tilde E$.

We want to demonstrate that for a.e. $\omega$, there exists some $N(\omega)$ such that $\mathds{1}_{A_n}(\omega)=0$ for $n \geq N(\omega)$.
\begin{align*}
I_X(T_{(n)}) - I_Y(T_{(n)}) &= \mathbb{E} \left( c (X^2 - Y^2) - \log \frac{|\partial_y g_n|}{|\partial_x g_n|} \:\middle\vert\: T_{(n)} \right)\\
&= \mathbb{E} \left(  X^2 - Y^2 \:\middle\vert\: T_{(n)} \right)\\
&\hspace{5mm}- \mathbb{E} \left( \log \frac{|\partial_y g_n|}{|\partial_x g_n|} \:\middle\vert\: T_{(n)} \right)
\end{align*}
It is clear that, for a.e. $T(\omega)$, the first term is uniformly bounded for all $n$.
Since $\log\frac{|\partial_y g_n|}{|\partial_x g_n|} \uparrow \infty$ a.e., it follows by the monotone convergence theorem for conditional expectation that
$$
\mathbb{E} \left( \log\frac{|\partial_y g_n|}{|\partial_x g_n|} \:\middle\vert\: T_{(n)} \right) \uparrow \infty \text{ a.s.}
$$
Thus, for almost every $\omega$, there is some $N_1(\omega)$ such that $I_X(T_{(n)}) \leq I_Y(T_{(n)})$ for $n \geq N_1(\omega)$, i.e. $\mathds{1}_{A_n}(\omega)=0.$
We may now apply the Dominated Convergence Theorem to the integral sequence (as it is monotonically decreasing a.e., by our assumption, it is bounded above by the integrable $\Umin{X}(T_{(1)}; X,Y)$)
\begin{align*}
\Umin{X}(T_{(n)}; X,Y) &= \mathbb{E} \left[ \mathds{1}_{A_n} (I_X(T_{(n)})  - I_Y(T_{(n)}) )
\right].
\end{align*}
Thus, we have the desired convergence $\Umin{X} \to 0$ as $n \to \infty$.

We demonstrate (\ref{eq:GenericKernel.limits.UminY}) similarly.
The expression for $\Umin{Y}$ given by Theorem~\ref{theorem:GenericKernel.UniqueInfo}, Eq.~(\ref{eq:GenericKernel.Thm.Umin.Gaussian}), can be written
\begin{equation*}
    \Umin{Y} =
\mathbb{E} \left[ \mathds{1}_{A_n} \left( \frac{\log(e)}{2} (Y^2 - X^2) + \log \frac{|\partial_y g| }{|\partial_x g| } \right) \right]
\end{equation*}
A similar argument to before may be used to demonstrate that $\mathds{1}_{A_n} \uparrow 1$ a.s. and thus
$$
\mathds{1}_{A_n} \left( \frac{\log(e)}{2} (Y^2 - X^2) + \log \frac{|\partial_y g| }{|\partial_x g| } \right) \uparrow \infty \text{ a.s.}
$$
It follows from the usual MCT that $\Umin{Y} \to \infty$ as $n \to \infty$.

We will now demonstrate Eqs.~(\ref{eq:GenericKernel.limits.UpmmX}) \& (\ref{eq:GenericKernel.limits.UpmX}).
From Theorem~\ref{theorem:GenericKernel.UniqueInfo}, Eq.~(\ref{eq:GenericKernel.Thm.Upm.Gaussian}), we have
\begin{equation}
    \Upm{X}(T_{(n)};X,Y)
        = \frac{\log(e)}{2} \sqrt{1-\rho^2} -
        \underbrace{\mathbb{E} \left[
        \mathds{1}_{C_n} \log \frac{|\partial_y g_{(n)}|}{|\partial_x g_{(n)}|}
         \right]}_{\Upmm{X}}
\end{equation}
Consider again the limit in Eq.~(\ref{eq:GenericKernel.limits.assumed_limit}), which defines our event $E$.
Recall that, by definition of the set $C_{(n)}$,
$$
\mathds{1}_{C_n} \log \frac{|\partial_y g_{(n)}|}{|\partial_x g_{(n)}|} = \left( \log \frac{|\partial_y g_{(n)}|}{|\partial_x g_{(n)}|} \right)^{+} .
$$
Thus, the limit from assumption implies that, for every $\omega \in \tilde{E}$,
$$
\left[
        \mathds{1}_{C_n} \log \frac{|\partial_y g_{(n)}|}{|\partial_x g_{(n)}|}
\right](\omega) \uparrow \infty .
$$
Since this random variable is non-negative (by the indicator), it follows from the Monotone Convergence Theorem that $\Upmm{X} \to \infty$, i.e. Eq.~(\ref{eq:GenericKernel.limits.UpmmX}).
Hence, Eq.~(\ref{eq:GenericKernel.limits.UpmX}) as well.

To demonstrate (\ref{eq:GenericKernel.limits.UpmmY}), consider the other direction:
\begin{equation}
\label{eq:GenericKernel.limits.proof.Upmy}
    \Upm{Y}(T_{(n)};X,Y)
        = \frac{\log(e)}{2} \sqrt{1-\rho^2} -
        \underbrace{\mathbb{E} \left[
        \mathds{1}_{C_n} \log \frac{|\partial_x g_{(n)}|}{|\partial_y g_{(n)}|}
         \right]}_{\Upmm{Y}}
\end{equation}
We have that for every $\omega \in \tilde{E}\cap C$, there is some $N_3(\omega)$ such that for
$$
\log \frac{|\partial_y g_{(n)}|}{|\partial_x g_{(n)}|}(\omega) < 0 \text{ for } n \geq N_3(\omega)
$$
and so $\mathds{1}_{C_n} = 0$.
Thus, using our assumption on $\Upm{Y}$, the DCT gives us (\ref{eq:GenericKernel.limits.UpmmY}), and (\ref{eq:GenericKernel.limits.UpmY}) follows from (\ref{eq:GenericKernel.limits.proof.Upmy}).
\end{proof}

With this result, we are now able to better explain the differences between the $\Imin$ and $\Ipm$ PID explored in Experiment III of Subsec.~\ref{subsection:experimentC}.
Recall that, for small $\alpha$, the $\Ipm$ PID assigned most of the joint predictor-interaction information $I(T;X,Y)$ to unique information $U_Y$ (Figs.~\ref{fig:experiment3},\ref{fig:supp.expIII}).
The $\Ipm$ PID, by contrast, assigned more information to redundancy and synergy, accounted for by negative $U_X$ (Fig.~\ref{fig:supp.expIII}).
Such a (continuous) sequence of sigmoidal kernels, parametrized $g_\alpha(x,y)$, corresponds to a (countable) sequence $g_n$, as in the statement of Corollary~\ref{corollary:GenericKernel.limits}.
Observe that
\begin{equation}
\frac{|\partial_x g_{\alpha}|}{|\partial_y g_{\alpha}|}(X,Y) 
= 
\frac{Y e^{\alpha - X}}{1+e^{\alpha - X}} \downarrow 0 \text{ as } \alpha \to -\infty \text{ almost surely}
\end{equation}
Thus, we may apply Cor.~\ref{corollary:GenericKernel.limits} to explain the differences between the $\Imin$ and $\Ipm$ PID observed in Sec.~\ref{section:SynergyNetworks}.
We will do so presently in Subsec.~\ref{subsection:GenericResults.applyCor2SigmoidalKernel}.

First, in order to better grasp the intuitive consequences of Theorem~\ref{theorem:GenericKernel.UniqueInfo} and Cor.~\ref{corollary:GenericKernel.limits}, we will begin with a simpler NFBI for which we can explicitly compute the full $\Imin$ and $\Ipm$ PIDs.

\subsection{Example: Continuous $\Imin$ and $\Ipm$ PIDs of Linear Interaction}
\label{subsection:GenericResults.linearPID}

We now present a fully worked out example of the PIDs of a continuous interaction, namely the simple linear interaction kernel that can be solved analytically without much difficulty (Theorem~\ref{theorem:linearIntxn}).
Moreover, this example will serve to clearly illustrate the implications of Corollary~\ref{corollary:GenericKernel.limits} for the difference between the $\Imin$ and $\Ipm$ PIDs.
In this subsection, we will take a high-level view.
In Supp.~\ref{subsection:appendix.computeLinearPIDs}, we provide a direct computation of these PIDs from the definitions in Sec.~\ref{section:preliminaries}.

These computations will prove quite instructive in elucidating the difference between these two measures --- that is, between $\Imin$ and $\Ipm$.
As noted in Sec.~\ref{section:previousWork}, our results in this section for the $\Imin$ PID overlap with those in \cite{Barrett2015}, although the approach is different.
In the following example, we will see how the continuous $\Imin$ induces a fully non-negative PID $\dmin$, while $\Ipm$ allows for negative atoms.
Crucially, the limiting properties of these PIDs will analytically support the specificity (resp., or lack thereof, resp.) in $\Smin$ (resp., of $\Spm$) that we observed in our experiments in Sec.~\ref{section:SynergyNetworks}.

Consider the following noiseless Gaussian interaction $X,Y \to_\rho T$:
\begin{align}
\label{eq:defn:linearIntxn}
    T &= aX + bY &\text{ where } 0 < a < b
\end{align}
It immediately follows that
\begin{align}
\label{eq:linearIntxn.TGaussian}
T &\sim N(0, \sigma_T^2)\\
\sigma_T^2 &= a^2 + b^2 + 2 \rho a b
\end{align}
Since this interaction is noiseless, we know that $I(T;X,Y) = \infty$ (Corollary \ref{corollary:infiniteMI.sumOfDiracs}). We affirm that $I(T;X)$ and $I(T;Y)$ are finite, however.
\begin{proposition}
\label{prop:linearIntxn.MI}
Let $X,Y \to_\rho T$ be the linear interaction as in Eq.~(\ref{eq:defn:linearIntxn}). Then the pairwise mutual informations between $X,Y,$ and $T$ are
\begin{align}
    \label{eq:linearIntxn.MI.XY}
    I(X;Y) &= - \log \sqrt{1-\rho^2}\\
    \label{eq:linearIntxn.MI.XT}
    I(T;X) &= - \log b + \log \sigma_T + I(X;Y) \\
    \label{eq:linearIntxn.MI.YT}
     I(T;Y) &= - \log a + \log \sigma_T + I(X;Y)
\end{align}
\end{proposition}
\begin{proof}
Eq.~(\ref{eq:linearIntxn.MI.XY}) is a standard result for jointly Gaussian variables, e.g. \cite[Ex.~8.5.1]{CoverThomas2005}.

Consider $(X,Y) \to (X,T)$ as defined by (\ref{eq:defn:linearIntxn}) as an invertible linear interaction from one bivariate Gaussian vector to another.
Per Rule~\ref{rule:Gaussian.linearTransformation} in the Appendix, we have that
\begin{align*}
\bm{\mu}_{X,T} &=  (0,0), \\
\Sigma_{X,T}
&=
\begin{bmatrix}
1 & a + \rho b \\
a + \rho b & \sigma_T^2
\end{bmatrix},\\
\sigma_T^2 &= a^2 + b^2 + 2 \rho a b .
\end{align*}
Moreover, we have the correlation coefficient:
\begin{align}
\label{eq:linearIntxn.rho.XT}
\rho_{X,T} &= \frac{a + \rho b}{\sigma_T} .
\end{align}
We use the same result for the MI of jointly Gaussian variables as before:
\begin{equation}
    \label{eq:MI.2Gaussians}
    I(X; T) = - \frac{1}{2} \log (1 - \rho_{X,T}^2).
\end{equation}
Using Eq.~(\ref{eq:linearIntxn.rho.XT}) and expanding, we arrive at Eq.~(\ref{eq:linearIntxn.MI.XT}).
The argument for Eq.~(\ref{eq:linearIntxn.MI.YT}) is identical, with $a$ and $b$ permuted in the expressions for $\Sigma_{Y,T}$ and $\rho_{Y,T}$.
\end{proof}

\begin{theorem}[PIDs $\dmin$ and $\dpm$ for Linear Interaction]
\label{theorem:linearIntxn}
Let $X,Y \to_\rho T$ be a linear interaction as in Eq.~(\ref{eq:defn:linearIntxn}). Then, using $I_{\cap}^{\min}$ as defined in Def.~\ref{defn:redundant-info.Imin}, we have the following PID $\dmin$:
\begin{subequations}
\begin{align}
    \label{eq:linearIntxn.WB.lattice.R}
    \Rmin &= - \log b + \log \sigma_T + I(X;Y) \\
     \nonumber &= \log \frac{\sqrt{a^2 + b^2 + 2 a b \rho}}{b \sqrt{1-\rho^2}} \\
    \label{eq:linearIntxn.WB.lattice.Ux}
    \Umin{X} &= 0, \\
    \label{eq:linearIntxn.WB.lattice.Uy}
    \Umin{Y} &= \log \frac{b}{a}, \\
    \label{eq:linearIntxn.WB.lattice.S}
    \Smin &= \infty,\\
    \intertext{where}
    \nonumber I(X;Y) &= - \log \sqrt{1-\rho^2}.
\end{align}
\end{subequations}
Using $I_{\PM}$ as defined in Def.~\ref{defn:redundant-info.Ipm}, we have the following PID $\dpm$:
\begin{subequations}
\begin{align}
    \label{eq:linearIntxn.PM.lattice.R}
    \Rpm
    &= -\log a + \log \sigma_T +I(X;Y) - \frac{1}{\pi} \sqrt{1-\rho^2} \\
    \label{eq:linearIntxn.PM.lattice.Ux}
    \Upm{X} &= \frac{1}{\pi} \sqrt{1-\rho^2} - \log \left( \frac{b}{a} \right) \\
    \label{eq:linearIntxn.PM.lattice.Uy}
    \Upm{Y} &= \frac{1}{\pi} \sqrt{1-\rho^2} \\
    \label{eq:linearIntxn.PM.lattice.S}
    \Spm &= \infty
\end{align}
\end{subequations}
\end{theorem}

We prove this Theorem in Supp.~\ref{subsection:appendix.computeLinearPIDs}.
Moreover, Eqs.~(\ref{eq:linearIntxn.WB.lattice.R}-\ref{eq:linearIntxn.WB.lattice.S}) are specific limiting case of the work in \cite{Barrett2015}.
There are two crucial differences that are immediately identifiable, however.
First, $\Rmin$ contains the term $-\log b$, as part of the difference $\log \sigma_T - \log b$. 
This difference shrinks to zero if $b$ grows relatively faster than $a$, as it does also when $a$ tends to zero faster than $b$.
In both cases limiting cases, $\Rmin$ reduces to $I(X,Y)$.
By contrast, for $a \to 0^+$ (faster than $b$), $\Rpm$ becomes arbitrarily large, on the order of $\log 1/a$.

As $a \to 0^+$, we approach an inadmissible extreme of the admissible kernels for Condition~\ref{condition} and our generic results, in which $X$ is becoming conditionally independent of $T$ given $Y$.
This is analagous to $\alpha \to -\infty$ for the non-linear sigmoidal kernel (\ref{eq:sigSw}) that we used for our simulations.
We may compute the $\Imin$ and $\Ipm$ PIDs of this limit by taking the limits of the expressions in Theorem~\ref{theorem:linearIntxn}.

\begin{corollary}
\label{corollary:linearIntxn.limits}
Let $\{X,Y \to_\rho T\}_a$ be the collection of linear interactions as in Eq~(\ref{eq:defn:linearIntxn}), where $b>0$ and $0 < \rho < 1$ are fixed and only $a\in(0,b)$ varies. 
Then as $a \to 0^+$, we have the following limits:
{\begin{subequations}
\label{eq:linearIntxn.limits.WB}
\begin{align}
\label{subeq:linearIntxn.limits.WB.R}
\Rmin &\to  I(X;Y) \\
\notag&= - \log \sqrt{1 - \rho^2}
\hspace{0.3in }\\
\label{subeq:linearIntxn.limits.WB.Ux}
\Umin{X} &= 0
\hspace{0.3in }\\
\label{subeq:linearIntxn.limits.WB.Uy}
\Umin{Y} &\to \infty
     \hspace{0.3in }\\
\label{subeq:linearIntxn.limits.WB.S}     
\Smin &= \infty 
    \hspace{0.3in }
\end{align}
\end{subequations}}
{\begin{subequations}
\label{eq:linearIntxn.limits.PM}
\begin{align}
\label{subeq:linearIntxn.limits.PM.R}
\Rpm & \to \infty \\
\label{subeq:linearIntxn.limits.PM.Ux}
\Upm{X} &\to - \infty \\
\label{subeq:linearIntxn.limits.PM.Uy}
\Upm{Y} &= \frac{1}{\pi} \sqrt{1-\rho^2} \\
\label{subeq:linearIntxn.limits.PM.S}
\Spm &= \infty
\end{align}
\end{subequations}}
where equality denotes a fixed value for all $a \in (0,b)$ (and hence trivial limit).
Moreover, the following ratios have limits:

{\begin{subequations}
\label{eq:linearIntxn.limitRatios.WB}
\begin{align}
\label{subeq:linearIntxn.limitRatios.WB.Uy}
\frac{\Umin{Y}}{I(T;Y)} &\to 1 \\
\label{subeq:linearIntxn.limitRatios.WB.R}
\frac{\Rmin}{I(T;Y)} &\to 0
\end{align}
\end{subequations}}
and
{\begin{subequations}
\label{eq:linearIntxn.limitRatios.PM}
\begin{align}
\label{subeq:linearIntxn.limitRatios.PM.Ux}
\frac{\Upm{X}}{I(T;Y)},\frac{\Upm{X}}{\Rpm} &\to -1\\
\label{subeq:linearIntxn.limitRatios.PM.R}
\frac{\Rpm}{I(T;Y)} &\to 1
\end{align}
\end{subequations}}
Combined, these provide
\begin{align}
\frac{\Upm{X}}{\Umin{Y}} &\to -1\\
\frac{\Umin{Y}}{\Rpm} &\to 1
\end{align}
\end{corollary}

\begin{proof}
The limits are readily computable with calculus, given the expressions in Theorem~\ref{theorem:linearIntxn}.
The ratio limits are most easily computed via disentangling the logarithms:
\begin{align*}
    \Upm{X} &= \log a + c_1(b,\rho)\\
    \Umin{Y} &= - \log a + c_2(b)\\
    \Rpm &=  - \log a + c_3(a,b,\rho)\\
    I(T;Y) &= - \log a + c_4(a,b,\rho)
\end{align*}
where the $c_i$ are either fixed terms of $b$ and $\rho$ or functions of $a,b,\rho$ that converge to a finite term as $a\to 0^+$.
\end{proof}

We will now interpret these limits according to our understanding.
We begin with the $\Imin$ PID.
As $a\to 0^+$, the redundant information $\Rmin$ approaches the total shared information between $X$ and $Y$ (\ref{subeq:linearIntxn.limits.WB.R}).
Since $T$ becomes asymptoptically determined by $Y$ (that is, $I(T:Y) \to \infty$) as $a \to 0^+$, this makes intuitive sense: 
any shared uncertainty between $X$ and $Y$ becomes $X$'s shared uncertainty with any variable exactly determined by $Y$.
It is also intuitive that $\Umin{Y} \to \infty$ (\ref{subeq:linearIntxn.limits.WB.Uy}).
Since $|\rho| < 1$, the mutual uncertainty $X$ shares with $Y$ (and thus $T$) is not all of either's uncertainty: $X$ and $Y$ share finite, rather than perfect, information of each other.
The gap from finite to perfect information is infinite, formally rendered $I(T;Y) - \Rmin(T;X,Y) = \infty$, and so $\Umin{Y} = \infty$ (\ref{subeq:linearIntxn.limits.WB.Uy}), i.e. if $Y$ has perfect information regarding $T$, then it must be unique to $Y$.
Less intuitive, on the other hand, is that $\Umin{X} = 0$ even when $a>0$ (\ref{subeq:linearIntxn.limits.WB.Ux}).
Though we might desire that $\Umin{X} \to 0$ as $a \to 0^+$, so long as $a>0$ and $|\rho| < 1$, there is some portion of uncertainty that $X$ shares with $T$ and not with $Y$.
This is another manifestation of the common critique of $\Imin$, which is that it does not distinguish the same information contained in the predictors from the same amount of information given by the predictors regarding a given outcome $T=t$. 
Since the expected total information that $X$ gives about any $T=t$ will always be less than that of $Y$ (as $a<b$), $\Imin$ will assign the full uncertainty that $X$ shares with $T$ to redundant information.

Let us now consider the $\Ipm$ PID $\dpm$.
To contextualize the limits of unique information, we note that this PID pseudobounds\footnote{In the discrete case, redundancy and ambiguity atoms are non-negative, and $\Upm{X}$ is indeed bound above by $\Upmp{X}$. This is no longer true in the continuous case, e.g. consider $|\rho| \to 1$ in Lemmas~\ref{lemma:linearIntxn.PM.lattice.specificity} \& \ref{lemma:linearIntxn.PM.lattice.ambiguity}.} $\Upm{X}$ and $\Upm{Y}$ above by the unique specificities $\Upmp{X}$ and $\Upmp{Y}$.
Since $Y$ has zero unique ambiguity, $\Upmm{Y} = 0$, we have that $\Upm{Y} = \Upmp{Y}$, a fixed quantity depending only upon the marginal distributions $p_X$ and $p_Y$.
We see that, as $a \to 0^+$ and $I(T;Y) \to \infty$, the $\Ipm$ PID treats the perfect information that $Y$ holds regarding $T$ as shared between the predictors (i.e., captured by the redundancy; see (\ref{subeq:linearIntxn.limitRatios.PM.R})), and it accounts for the finite information in $I(T;X)$, which approaches zero, as what we may think of as an \textit{infinite shortcoming} of $X$ as a predictor (see (\ref{subeq:linearIntxn.limitRatios.PM.Ux})).
Explicitly, for the expression for $\Upm{X}$ in Eq.~(\ref{eq:linearIntxn.PM.lattice.Ux}), this shortcoming is quantified with the term $-\log\frac{b}{a}$, which grows without bound in magnitude in proportion to the discrepancy between the information in $X$ and $Y$.
This is the exactly the opposite term by which $\Umin{Y} \to \infty$ for the $\Imin$ PID in Eq.~(\ref{eq:linearIntxn.WB.lattice.Uy}).

The allocation of this term, $\log (b/a)$, distinguishes the $\Imin$ limit from the $\Ipm$ limit, i.e. whether $U_Y \to \infty$ or $U_X \to -\infty$ as $X$ becomes conditionally independent of $T$ given $Y$.
As to the significance of the term $\log(b/a)$, observe that this is the target MI difference between the predictors:
\begin{equation}
\label{eq:MI_discrep.linearIntxn}
    I(T;Y) - I(T;X) = \log(b/a).
\end{equation}
This difference becomes infinite as $a \to 0^+$ because, at that point, the difference is one of perfect versus finite information, whereas, when $b>a>0$, only $(X,Y)$ together have perfect information regarding $T$.
In the more general case of an NFBI with identical predictor marginals $p_X(X) = p_Y(X)$, this difference takes the form
\begin{equation}
\label{eq:MI_discrep.genericKernel}
    I(T;Y) - I(T;X) = \mathbb{E} \log \frac{|\partial_y g|}{|\partial_x g|}(X,Y).
\end{equation}
A key distinction, then, between the bivariate $\Imin$ and $\Ipm$ PIDs is in the allocation of this information discrepancy (\ref{eq:MI_discrep.linearIntxn}-\ref{eq:MI_discrep.genericKernel}), which becomes infinite as $Y$ becomes the sole determinant of $T$.
The $\Imin$ PID will assign it to $\Umin{Y}$, while the $\Ipm$ will assign it to $\Rpm$, with a corresponding negative atom in $\Upmm{X}$.

Now that we have provided both PIDs of the linear kernel and demonstrated their limits, we may turn our attention to applying our general results to the sigmoidal kernel (\ref{eq:sigSw}).

\subsection{Asymptotic Conditional Independence for Sigmoidal Interactions}
\label{subsection:GenericResults.applyCor2SigmoidalKernel}

We will conclude this section with an analytic investigation into noise-free bivariate interactions with the sigmoidal switch interaction kernel (\ref{eq:sigSw}), as used for the simulations in Sec.~\ref{section:SynergyNetworks}.
This subsection will demonstrate the application of the generic approach that we have been developing throughout this section, in Theorem~\ref{theorem:GenericKernel.UniqueInfo} and Corollary~\ref{corollary:GenericKernel.limits}.
Our application of Cor.~\ref{corollary:GenericKernel.limits} will demonstrate the non-specificity of the $\Ipm$ PID, in contrast to the $\Imin$ PID.

In Experiment~III in Sec.~\ref{subsection:experimentC}, we conduct a series of network simulations with the sigmoidal kernel, in which we allowed the switching parameter $\alpha$ to vary.
We then examine the effect this had upon the PID of interacting pairs of genes.
Altering $\alpha$ has the effect of re-centering the predictor $X$ about $X=\alpha$ in the interaction.
This has the effect of shifting the relative sensitivity of $g(X,Y)$ upon each argument $X$ and $Y$ near the mean $\mu_X = \mu_Y = 0$, and thus in expectation as well.
More precisely, we saw that the expansion of the sigmoidal kernel takes the form:
\begin{align}
\tag{\ref{eq:sigSw.2ndOrderExpansion}}
    g(x,y) \approx \underbrace{\frac{1}{1 + e^{\alpha}}}_{\partial_{y} g} y + \underbrace{\frac{e^{\alpha}}{(1 + e^{\alpha})^2}}_{\partial_{x,y} g} xy .
\end{align}
for $(x,y) \approx (0,0)$, which is the probable scenario.
Thus, as $\alpha \to -\infty$, we have that $g(x,y) \approx y$ and $g(X,Y)$ approaches conditional independence of $X$ given $Y$.
We noted that $\Upm{X}$ was consistently negative for a range of $\alpha$ (Supp. Fig.~\ref{fig:supp.expIII}). 
For a linear kernel, we saw in Sec.~\ref{subsection:GenericResults.linearPID} that $\Upm{X} \to - \infty$ as $X$ approaches conditional independence.
A core contention of ours in this work is that, in such a scenario, we ought to see unique information $U_X$ approach zero as $X$ becomes less conditionally dependent on $T$, as otherwise the bivariate PID will inflate both redundancy and synergy scores for spurious edges during network inference.
The $\Imin$ PID behaves in this way, as seen in Sec.~\ref{section:SynergyNetworks}, but the $\Ipm$ PID does not.

Using Corollary~\ref{corollary:GenericKernel.limits}, we may now demonstrate that, similar to the linear interaction limits in Cor.~\ref{corollary:linearIntxn.limits}, we have that $\Umin{X} \to 0$ and $\Upm{X} \to -\infty$ as $\alpha \to -\infty$.
Thus, the themes explored regarding how the $\Imin$ and $\Ipm$ PIDs treat conditionally independent, `false' interactions can be made analytically apparent in the sigmoidal interaction kernel as they did for the linear interaction kernel in Sec.~\ref{subsection:GenericResults.linearPID}.
Moreover, since we are now examining the sigmoidal kernel, the work in this section can better explain the behavior we saw in our experiments in Sec.~\ref{section:SynergyNetworks}.

\begin{proposition}
\label{prop:sigSw.limits}
Let $\{X, Y \to_\rho T_{(\alpha)} \}_{\alpha \in \mathbb{R}}$ be the family of noise-free Gaussian interactions with kernel $g_{(\alpha)}$ from Eq.~\ref{eq:sigSw}, parametrized by the real parameter $\alpha$.

Then we have the following limits:
\begin{align}
\Umin{X}(\alpha) &\to 0 &\text{ as } \alpha \to -\infty\\
\Upmm{X}(\alpha) &\to -\infty &\text{ as } \alpha \to \infty\\
\Upm{X}(\alpha) &\to -\infty &\text{ as } \alpha \to -\infty
\end{align}
\end{proposition}

\begin{proof}
We may apply Corollary~\ref{corollary:GenericKernel.limits} to our continuously parametrized family of interactions via the sequential characterization of continuous limits.
Observe that
\begin{align*}
    \frac{|\partial_x g|}{|\partial_y g|} &= \frac{|Y| e^{\alpha-X}}{1 + e^{\alpha-X}}
\end{align*}
It is easy to see that for any $(x,y) = (X,Y)(\omega)$, $\alpha \leq 0$, 
\begin{align*}
    \frac{|\partial_x g|}{|\partial_y g|} \downarrow 0 \text{ as } \alpha \to - \infty.
\end{align*}
We take $\alpha=0$ as our starting point, so that all that remains to be shown is that $\Umin{X}(0)$ and $\Upm{X}(0)$ are finite.

Consider the following series of expansions
\begin{align*}
    &\frac{1}{2}|x^2 - y^2| + \left| \ln \frac{|\partial_x g|}{|\partial_y g|} \right|
    \\
    =& \frac{1}{2}|x^2 - y^2| + \left| \ln|y| + (\alpha - x) - \log(1 + e^{\alpha - x}) \right| \\
    \leq&  \frac{1}{2}|x^2 - y^2| + \left|\ln|y| \right| + \alpha + |x| + \overbrace{\log(1 + e^{\alpha - x})}^{\leq e^{\alpha-x}} \\
    \leq& \frac{x^2}{2} + \frac{y^2}{2} + \alpha + |x| + \left|\ln|y| \right| + e^{\alpha-x}
\end{align*}

Set $\alpha = 0$, and consider the functions $f_1$ and $f_2$:
\begin{align*}
    f_1(x,y) &:= \frac{1}{2}|x^2 - y^2| + \left| \ln \frac{|\partial_x g|}{|\partial_y g|} \right|\\
    &\leq \frac{x^2}{2} + \frac{y^2}{2} + \alpha + |x| + \left|\ln|y| \right| + e^{\alpha-x}
    =: f_2(x,y)
\end{align*}

All the terms in $f_2$ have finite moments with respect to the marginal
$$
p_X(z) = p_Y(z) = \frac{1}{\sqrt{2 \pi}} e^{-\frac{z^2}{2}}
$$
and so $\mathbb{E} f_1(X,Y) \leq \mathbb{E} f_2(X,Y) < \infty$.

Moreover, from the forms of $\Umin{X}$ and $\Upm{X}$ in Theorem~\ref{theorem:GenericKernel.UniqueInfo}, we have that
\begin{align*}
    |\Umin{X}(0)| \leq \mathbb{E} f_1(X,Y),\\
    |\Upm{X}(0)| \leq \mathbb{E} f_1(X,Y)+c,
\end{align*}
for a finite constant $c>0$.
Thus, $\Umin{X}(0)$ and $\Upm{X}(0)$ are finite.
\end{proof}

\section{Concluding Thoughts}
\label{section:conclusion}
In this work, we have considered the $\Imin$ and $\Ipm$ PID frameworks as tools for edge nomination in response/interaction networks.
We take a multifaceted approach, simulating interaction networks in Section~\ref{section:SynergyNetworks} in order to build intuition, before formally investigating noise-free bivariate interactions in the general setting of Section~\ref{section:GenericResults}.
Our investigations have demonstrated that the $\Ipm$ PID is ill-suited to the task of edge/interaction nomination in synergy networks.
Analytically and in simulation, we see that the bivariate $\Ipm$ PID suffers from a non-specific sensitivity to any pair of high-information predictors, and does not distinguish between true interactions and strong univariate signals.
This is also a drawback of classical pair-wise MI.
By contrast, the $\Imin$ PID demonstrates behavior more appropriate to the edge nomination task, in that
$\Imin$ properly distinguishes true interactions from univariate signals.
This is accomplished by shifting its allocation of information in the bivariate PID in accordance with the relative analytic sensitivity of the response variable to the predictors.
Crucially, the $\Imin$ PID respects the conditional independence of predictors (Prop.~\ref{prop:PID.NN_implies_zero_synergy}), and this is further demonstrated in the asymptotic analysis of Cor.~\ref{corollary:GenericKernel.limits} for noise-free interactions.

In future work, we would conduct a similar analysis of PIDs besides $\Imin$ and $\Ipm$.
Both the continuous $\Ibroja$ PID from \cite{Pakman2021estimating} and the $\Isx$ PID from \cite{Schick2021partial} ought to be considered from a similar perspective.
From the simulations conducted in \cite{Pakman2021estimating}, particularly for neuronal Model 2, we suspect that the continuous $\Ibroja$ would resemble the $\Imin$ PID in demonstrating specificity without sufficient sensitivity to synergistic interactions.
This would align with the earlier work from \cite{Barrett2015} suggesting that both are identical to the MMI PID for Gaussian systems.
The $\Isx$ PID, on the other hand, is akin to the $\Ipm$ PID, and thus may similarly apportion negative information to $U_X$ for an irrelevant predictor $X$, as $\Ipm$ did in our work.
Besides these existing continuous PIDs, we also note that the $\Iccs$ PID performed well in our experiment.
To our knowledge, the $\Iccs$ PID has not been examined in continuous variables.
Since $\Iccs$ does not obey the monotonicity WB axiom, and allows for negative information, it would be interesting to discover whether this PID avoids any of the drawbacks of the $\Ipm$ PID that we have explored in this work.

\vfill

\nocite{*}

\bibliography{refs}

\clearpage

\section{Supplemental Material}
\beginsupplement

 \subsection{Computation of Continuous $\Imin$ and $\Ipm$ PIDs for Linear Interactions}
\label{subsection:appendix.computeLinearPIDs}

In this section, we compute the $\Imin$ and $\Ipm$ PIDs in Theorem~\ref{theorem:linearIntxn}, for the continuous interaction in Eq.~(\ref{eq:defn:linearIntxn}).
The structure of this section is as follows.
First, we restate both PIDs in Theorem~\ref{theorem:linearIntxn}.
Next, since the $\Imin$ PID depends upon the computation of the specific information functions $I_X(t)$ and $I_Y(t)$, we present these in Lemma~\ref{lemma:linearIntxn.WB.specificInfo}, followed by its proof.
This allows us to prove the $\Imin$ PID in Theorem~\ref{theorem:linearIntxn}, which follows immediately from Lemma~\ref{lemma:linearIntxn.WB.specificInfo} and the MIs from Prop.~\ref{prop:linearIntxn.MI}.
It is worth highlighting that this computation is particularly simple due to the convenient property that $I_Y(t) \geq I_X(t)$ for all $t$, which is a necessary and sufficient condition for the coincidence of $\Imin$ redundancy and minimal predictor MI, i.e. $I(T;X) = \Rmin(T;X,Y)$.
This is synonymous with the MMI PID from \cite{Barrett2015}.
We then turn to the $\Ipm$ PID.
Since this PID is the difference of the decomposition of specificity $I^+(T;X,Y)$ and ambiguity $I^-(T;X,Y)$, we compute each of those lattices in turn, in Lemmas~\ref{lemma:linearIntxn.PM.lattice.specificity}~\&~\ref{lemma:linearIntxn.PM.lattice.ambiguity}.
With those proven, we conclude the computation of our PIDs.

\begin{claim}[PID $\dmin$ for Linear Interaction]
\label{claim:linearIntxn.WB.lattice}
Let $X,Y \to_\rho T$ be a linear interaction as in Eq.~(\ref{eq:defn:linearIntxn}). Then, using $I_{\cap}^{\min}$ as defined in Def.~\ref{defn:redundant-info.Imin}, it has the following PID $\dmin$:
\begin{align}
    \tag{\ref{eq:linearIntxn.WB.lattice.R}}
    R &= - \log b + \log \sigma_T + I(X;Y) \\
    \nonumber &= \log \frac{\sqrt{a^2 + b^2 + 2 a b \rho}}{b \sqrt{1-\rho^2}} \\
    \tag{\ref{eq:linearIntxn.WB.lattice.Ux}}
    U_X &= 0 \\
    \tag{\ref{eq:linearIntxn.WB.lattice.Uy}}
    U_Y &= \log \frac{b}{a} \\
    S &= \infty
    \tag{\ref{eq:linearIntxn.WB.lattice.S}}
\end{align}
\end{claim}

\begin{claim}[PID $\dpm$ for Linear Interaction]
\label{claim:linearIntxn.PM.lattice}
Let $X,Y \to_\rho T$ be a linear interaction as in Eq.~(\ref{eq:defn:linearIntxn}). Then, using $I_{\PM}$ as defined in Def.~\ref{defn:redundant-info.Ipm}, it has the following PID $\dpm = \mathfrak{d}^+_{\PM} - \mathfrak{d}^-_{\PM}$:
\begin{align}
    \tag{\ref{eq:linearIntxn.PM.lattice.R}}
    R &= \log \frac{\sqrt{a^2 + b^2 + 2 a b \rho}}{a \sqrt{1-\rho^2}} - \frac{1}{\pi} \sqrt{1-\rho^2} \\
    \tag{\ref{eq:linearIntxn.PM.lattice.Ux}}
    U_X &= \frac{1}{\pi} \sqrt{1-\rho^2} - \log \left( \frac{b}{a} \right) \\
    \tag{\ref{eq:linearIntxn.PM.lattice.Uy}}
    U_Y &= \frac{1}{\pi} \sqrt{1-\rho^2} \\
    S &= \infty
    \tag{\ref{eq:linearIntxn.PM.lattice.S}}
\end{align}
\end{claim}

We begin with the $\Imin$ lattice $\dmin$.
In order to demonstrate Claim~\ref{claim:linearIntxn.WB.lattice}, we need to first compute the specific information functions for each predictor variable.

\begin{lemma}[Specific Information Functions for Linear Interaction]
\label{lemma:linearIntxn.WB.specificInfo}
Let $X,Y \to_\rho T$ be a linear interaction as in Eq.~(\ref{eq:defn:linearIntxn}). Then the specific information functions $I_X(t)$ and $I_Y(t)$ as given in Def.~\ref{defn:specificInfo} are given as the following:
\begin{subequations}
\label{eqns:linearIntxn.WB.specificInfo}
\begin{align}
    I_X(t) &= - \log b + \frac{1}{2 \sigma_T^4}
\Big[ b^2 (1 - \rho^2) \sigma_T^2\notag\\ 
&\hspace{20mm}+ (a + \rho b)^2 t^2 \Big] + c \label{eqn:linearIntxn.WB.specificInfo.Ix} \\
    I_Y(t) &= - \log a + \frac{1}{2 \sigma_T^4}
\Big[ a^2 (1 - \rho^2) \sigma_T^2\notag\\
&\hspace{20mm}+ (b + \rho a)^2 t^2 \Big]  + c \label{eqn:linearIntxn.WB.specificInfo.Iy}\\
    \intertext{where}
    \label{eqn:linearIntxn.WB.specificInfo.constant}
    c &= \log \sigma_T - \log \sqrt{1-\rho^2} - \frac{1}{2} \\
    \label{eqn:linearIntxn.WB.specificInfo.sigmaT}
    \sigma_T^2 &= a^2 + b^2 + 2 \rho a b
    \intertext{Moreover, their difference $I_Y - I_X$ is positive, given by:}
     \label{eqn:linearIntxn.WB.specificInfo.difference}
    I_Y(t) - I_X(t) &= \log \frac{b}{a}
    + \frac{1}{2 \sigma_T^4}
    \left[
    (1-\rho^2)(b^2 - a^2) (t^2 - \sigma_T^2)
    \right]
\end{align}

\end{subequations}

\end{lemma}

\begin{proof}[Proof of Lemma~\ref{lemma:linearIntxn.WB.specificInfo}]
As in the proof of Prop~\ref{prop:linearIntxn.MI}, we have that
\begin{align*}
\bm{\mu}_{X,T} &=  (0,0), \\
\Sigma_{X,T}
&=
\begin{bmatrix}
1 & a + \rho b \\
a + \rho b & \sigma_T^2
\end{bmatrix},\\
\sigma_T^2 &= a^2 + b^2 + 2 \rho a b, \\
\rho_{X,T} &= \frac{a + \rho b}{\sigma_T}.
\end{align*}
Using Rule~\ref{rule:conditionalGaussian}, this gives us a conditional distribution $X_{T=t} = (X | {T=t})$ for any choice of $t$:
\begin{align}
\nonumber X_{T=t} &\sim N(\mu_{X|{T=t}}, \sigma^2_{X|T=t} )\\
\intertext{where}
\label{eqn:proof:linearIntxn.WB.specificInfo.ref1}
\mu_{X|T=t} &= \frac{ (a + \rho b) t}{\sigma_T^2} \tag{$*$}\\
\label{eqn:proof:linearIntxn.WB.specificInfo.ref2}
\sigma^2_{X|T=t} &= \frac{b^2 (1 - \rho^2)}{\sigma_T^2} \tag{$**$}
\end{align}

This allows us to compute the specific information at $T=t$ as the KL divergence between $X |_{T=t}$ and $X$. Both are univariate Gaussians, so we use Rule~\ref{rule:KLdivergence.twoGaussians}, and have that
\begin{align*}
I_X(t) &= D( X_{T=t} || X)\\
&= \log \left( \frac{1}{\sigma_{X|T=t}} \right) +
\frac{1}{2} \left( 
\sigma^2_{X|T=t} + \mu^2_{X|T=t}
  \right) - \frac{1}{2} .
\end{align*}
We plug in Eqs~(\ref{eqn:proof:linearIntxn.WB.specificInfo.ref1}) and (\ref{eqn:proof:linearIntxn.WB.specificInfo.ref2}) to arrive at the form in Eq.~\ref{eqn:linearIntxn.WB.specificInfo.Ix}. In particular, with $c$ as in Eq.~\ref{eqn:linearIntxn.WB.specificInfo.constant}, note that
\begin{align*}
    \log \left( \frac{1}{\sigma_{X|T=t}} \right) - \frac{1}{2}
    &= - \log b + c,\\
\sigma^2_{X|T=t} + \mu^2_{X|T=t}
    &= \frac{1}{\sigma_T^4} \left[ b^2 (1 - \rho^2)\sigma_T^2 + (a + \rho b)^2 t^2 \right].
\end{align*}

The proof for Eq.~\ref{eqn:linearIntxn.WB.specificInfo.Iy} is identical under the transposition of $a$ and $b$.

From Eq.~\ref{eqn:linearIntxn.WB.specificInfo.Ix} and Eq.~\ref{eqn:linearIntxn.WB.specificInfo.Iy}, we have that Eq.~\ref{eqn:linearIntxn.WB.specificInfo.difference} follows by simply expanding terms and cancelling. To demonstrate that this difference is strictly positive, it suffices to consider the worst-case scenario of (\ref{eqn:linearIntxn.WB.specificInfo.difference}) where $t=0$:
\begin{align*}
      I_Y(0) - I_X(0) &= \log \frac{b}{a}
    - \frac{1}{2 \sigma_T^2}
    \left[
    (1-\rho^2)(b^2 - a^2)
    \right]
\end{align*}

We make the substitution $b = \gamma a$, and the above becomes a function $f$ of $\gamma$ (and $\rho$, implicitly):

\begin{align}
\label{eqn:proof:linearIntxn.WB.specificInfo.specialFunction}
     f(\gamma) &:=  \log \gamma
- \frac{1 (1-\rho^2)(\gamma^2 - 1)}{2 (\gamma^2 + 2 \rho \gamma + 1)}
\end{align}

Note that $\gamma > 1$ since $0 < a < b$. In Computation~\ref{computation.linearIntxn.WB.specificInfo.differenceFunctionIncreasing}, we demonstrate this function is strictly increasing for $\gamma \geq 1$, under the assumption $\rho \in [-1,1]$. Since $f(1) = 0$, it follows that $f(\gamma)>0$ for any $\gamma>1$ and $\rho \in [-1,1]$. 
\end{proof}

We are now ready to prove Claim \ref{claim:linearIntxn.WB.lattice}, as part of Theorem~\ref{theorem:linearIntxn}.  
\begin{proof}[Proof of Claim  \ref{claim:linearIntxn.WB.lattice}]
The form of $\Imin$ follows from $I_X(t)\leq I_Y(t)$ for all $t$ and hence $\Imin=\mathbb{E}(I_X(T))$ where $I_X(T)$ has the form of Eq.\@ (\ref{eqn:linearIntxn.WB.specificInfo.constant}).
It is immediate that $\Umin{X}=0$ and the form of $\Umin{Y}$ follows from combining equations Eq.\@ (\ref{eq:linearIntxn.MI.YT}) and Eq.\@ (\ref{eqn:linearIntxn.WB.specificInfo.constant}).
Lastly, $\Smin=\infty$ follows from $I(T;X,Y)=\infty$ and Eq.\@ (\ref{eq:PID_1}).
\end{proof}

We now turn to the proof of the $\Ipm$ PID for Theorem~\ref{theorem:linearIntxn}, i.e. Claim~\ref{claim:linearIntxn.PM.lattice}.
This PID follows readily from computing the specificity and ambiguity lattices, which we do in Lemmas~\ref{lemma:linearIntxn.WB.specificInfo} and \ref{lemma:linearIntxn.PM.lattice.ambiguity}, respectively.
The specificity lattice does not depend on the interaction itself, only the distribution of the predictors $X$ and $Y$. Thus, we also use Lemma~\ref{lemma:linearIntxn.WB.specificInfo} for NFBIs with similarly distributed predictors, as we do in Theorem~\ref{theorem:GenericKernel.UniqueInfo}.

\begin{lemma}[$I_{\PM}$ Specificity Lattice for Gaussian PID]
\label{lemma:linearIntxn.PM.lattice.specificity}
Let $X,Y \to_\rho T$ be a Gaussian bivariate interaction, noisy or noiseless. Then the specificity lattice $\mathfrak{d}^+_{\PM}$ as defined by Def.~\ref{defn:redundant-info.Ipm} has the following form:
\begin{align}
    R^+ &= \log(\sqrt{2 \pi e}) - \frac{1}{\pi} \sqrt{1-\rho^2} \\
    U_X^+ &= \frac{1}{\pi} \sqrt{1-\rho^2} \\
    U_Y^+ &= \frac{1}{\pi} \sqrt{1-\rho^2} \\
    S^+ &= \log(\sqrt{2 \pi e}) + \log(\sqrt{1 - \rho^2}) - \frac{1}{\pi} \sqrt{1-\rho^2}
\end{align}
\end{lemma}

\begin{proof}

We begin by computing the redundant specificity
\begin{align*}
R^+ &= \mathbb{E} r^+(X,Y)\\
\intertext{where}
r^+(x,y) &= \min (-\log(p_X(x)), - \log(p_Y(y))).
\end{align*}
We take the expectation over $X \times Y$ space, $(\mathbb{R}^2, \mu_{X \times Y})$. Since $p_X = p_Y \sim N(0,1)$, we have that
\begin{align*}
p_X(u) = p_Y(u) = \frac{1}{\sqrt{2\pi}} e^{\frac{-u^2}{2}}
\intertext{and so }
    - \log p_X(x) \leq - \log p_Y(y) \Leftrightarrow x^2 \leq y^2.
\end{align*}
This yields that (where we use the fact that $2\min(x,y)=x+y-|x-y|$
\begin{align*}
\mathbb{E} r^+(X,Y)&=\log(\sqrt{2\pi})+\frac{1}{2}\mathbb{E}(\min(X^2,Y^2))\\
&=\log(\sqrt{2\pi})+\frac{1}{4}\left(\mathbb{E}X^2+\mathbb{E}Y^2-\mathbb{E}(|X^2-Y^2|)\right)\\
&=\log(\sqrt{2\pi})+\frac{1}{4}\left(2-\mathbb{E}(|X-Y|)\mathbb{E}(|X+Y|)\right)\\
&=\log(\sqrt{2\pi})+\frac{1}{4}\left(2-\frac{2}{\pi}\sqrt{4-4\rho^2}\right)\\
&=\log(\sqrt{2\pi})+\frac{1}{2}-\frac{1}{\pi}\sqrt{1-\rho^2},
\end{align*}
where we used above the fact that $X-Y\sim N(0,2-2\rho)$ is independent of $X+Y\sim N(0,2+2\rho)$ (indeed, the covariance is Cov$(X-Y,X+Y)=0$), and for $W\sim$Norm$(0,\sigma^2)$, $\mathbb{E}(|W|)=\sigma\sqrt{2/\pi}$.

We now turn to the unique specificities $U^+_X$ and $U^+_Y$. The specificity components $I^+(T;X)$ and $I^+(T;Y)$ of the mutual informations are just the entropies $h(X)$ and $h(Y)$, by definition (Def.~\ref{defn:redundant-info.Ipm}).
Since $X$ and $Y$ are standard normal variables in their marginal distributions, we have that
$$
I^+(T;X) = I^+(T;Y) = \log(\sqrt{2 \pi e}).
$$
Thus, unique specificity becomes
\begin{align*}
    U_X^+ = U_Y^+ &= I^+(T;Y) - R^+\\
    &= \frac{1}{\pi} \sqrt{1-\rho^2}
\end{align*}

Finally, we compute synergistic specificity. Again, the specificity component of the mutual information $I(T;X,Y)$ is the joint entropy:
$$
I^+(T;\{X,Y\}) = h(X,Y) = \log(2 \pi e) + \log(\sqrt{1 - \rho^2}).
$$
So, by Eq.~(\ref{eq:redundant-info.Ipm:eqn:sublattice.S}), we have that
\begin{align*}
    S^+ &= \log(\sqrt{2 \pi e}) + \log(\sqrt{1 - \rho^2}) - \frac{1}{\pi} \sqrt{1-\rho^2}
\end{align*}
\end{proof}

Whereas the specificity lattice depends only upon the distribution of the predictors, the ambiguity lattice depends upon the actual interaction. This is also the component where the synergistic component blows up for a noiseless interaction.

\begin{lemma}[$I_{\PM}$ Ambiguity Lattice for Linear Interaction]
\label{lemma:linearIntxn.PM.lattice.ambiguity}
Let $X,Y \to_\rho T$ be a linear interaction as in Eq.~(\ref{eq:defn:linearIntxn}). Then the ambiguity lattice $\mathfrak{d}^-_{\PM}$, as defined by Def.~\ref{defn:redundant-info.Ipm} has the following form:
\begin{align}
    R^- &= \log(\sqrt{2 \pi e}) + \log \sqrt{1 - \rho^2} + \log \left( \frac{a}{\sqrt{a^2 + b^2 + 2 a b \rho}} \right) \\
    U_X^- &= \log \left( \frac{b}{a} \right) \\
    U_Y^- &= 0 \\
    S^- &= - \infty
\end{align}
\end{lemma}

\begin{proof}
We begin with redundant ambiguity. A change of variables reveals that, for a given allowable triplet $(X,Y,T) = (x,y,t)$:
\begin{align*}
    \label{eqn:proof:linearIntxn-Ambig.ref1}
    \frac{p_{X,T}(x,t)}{p_T(t)}&= \frac{1}{b}\frac{\frac{}{} p_{X,Y}(x,y)}{p_T(t)} \tag{$*$} \\
    \label{eqn:proof:linearIntxn-Ambig.ref2}
    \frac{p_{Y,T}(y,t)}{p_T(t)} &= \frac{1}{a}\frac{\frac{}{} p_{X,Y}(x,y)}{p_T(t)} \tag{$**$}
\end{align*}
We see that, since $0 < a < b$, we have that
\begin{align}
    \nonumber - \log \frac{p_{Y,T}(y,t)}{p_T(t)} &\leq - \log \frac{p_{X,T}(x,t)}{p_T(t)} \\
    \intertext{and so, using Def.~\ref{defn:redundant-info.Ipm}:}
    \label{eqn:proof:linearIntxn-Ambig.ref3}
    R^- &= \mathbb{E} \left( - \log \frac{p_{Y,T}(y,t)}{p_T(t)} \right) \tag{$***$} \\
    \nonumber &= \log a + h(X,Y) - h(T)
\end{align}
Using standard formulae for the entropy of Gaussians \cite[ch.~7]{CoverThomas2005}, as well as Eq.~\ref{eq:linearIntxn.TGaussian}, we have that
\begin{align*}
    h(T) &=  \log \sqrt{2 \pi e} + \log(\sigma_T) \\
    h(X,Y) &= \log(2 \pi e) + \log(\sqrt{1-\rho^2})
\end{align*}
Thus, redundant ambiguity becomes
$$
R^- = \log \sqrt{2 \pi e} + \log \sqrt{1-\rho^2} + \log \left( \frac{a}{a^2 + b^2 + 2 \rho a b} \right).
$$

We can now turn to the unique ambiguity atoms $U_X^-$ and $U_Y^-$. By Def.~\ref{defn:redundant-info.Ipm}, Eqs~(\ref{eq:redundant-info.Ipm:eqn:sublattice.Ux}) and (\ref{eq:redundant-info.Ipm:eqn:sublattice.Uy}), we have that
\begin{align*}
    U_X^- &= \mathbb{E} \left( - \log \frac{p_{X,T}(x,t)}{p_T(t)} \right) - R^- \\
    U_Y^- &= \mathbb{E} \left( - \log \frac{p_{Y,T}(y,t)}{p_T(t)} \right) - R^-
\end{align*}
By plugging in Eq.~(\ref{eqn:proof:linearIntxn-Ambig.ref3}), we arrive at
\begin{align*}
    U_Y^- &= 0\\
    U_X^- &= \mathbb{E} \left( - \log \frac{p_{X,T}(x,t)}{p_T(t)} \right) - \mathbb{E} \left( - \log \frac{p_{Y,T}(y,t)}{p_T(t)} \right) \\
\end{align*}
We expand $U_X^-$ with Eqs~(\ref{eqn:proof:linearIntxn-Ambig.ref1}) and (\ref{eqn:proof:linearIntxn-Ambig.ref2}):
\begin{align*}
    U_X^- &= \mathbb{E} \left( - \log \frac{p_{X,T}(x,t)}{p_T(t)} \right) - \mathbb{E} \left( - \log \frac{p_{Y,T}(y,t)}{p_T(t)} \right) \\
    &= (\log b + h(X,Y) - h(T)) - (\log a + h(X,Y) - h(T)) \\
    &= \log \left( \frac{b}{a} \right)
\end{align*}

We have computed finite values for the atoms $R^-, U_X^-,$ and $U_Y^-$. Since $I(T;X,Y) = \infty$, we note that $I^{-}(T;X,Y) = - \infty$, and thus Eq.~(\ref{eq:redundant-info.Ipm:eqn:sublattice.S}), when rewritten in the following form:
$$
I^{-} = R^{-} + S^{-} + U_X^{-} + U_Y^{-}
$$
makes it determinate that $S^{-} = -\infty$.

\end{proof}

Combining Lemmas~\ref{lemma:linearIntxn.PM.lattice.specificity} \& \ref{lemma:linearIntxn.PM.lattice.ambiguity}, we arrive at the PID in Claim~\ref{claim:linearIntxn.PM.lattice}.
Thus, we have the exact values of the continuous PIDs $\dmin$ and $\dpm$ for the linear interaction in Theorem~\ref{theorem:linearIntxn}.

\subsection{Commonplace Rules and Auxiliary Computations}
\label{subsection:appendix.rulesAndComputations}

\subsubsection{Rules for Gaussian Interactions}
\begin{newrule}
\label{rule:Gaussian.linearTransformation}
 In general, if $\bm{U} \sim N(\bm{\mu}_{\bm{U}}, \Sigma_{\bm{U}})$ and $\bm{V} \sim N(\bm{\mu}_{\bm{V}}, \Sigma_{\bm{V}})$ are $k$-dimensional Gaussian vectors such that $\bm{V} = A \bm{U}$ for an $n \times n$ real, nonsingular transformation $A$, then
\begin{align*}
    \bm{\mu}_V &= A \bm{\mu}_U, \\
    \Sigma_{\bm{V}} &= A \Sigma_{\bm{U}} A^T.
\end{align*}
\end{newrule}

\begin{newrule}
\label{rule:conditionalGaussian}
Let $X_1, X_2$ be two Gaussians with correlation $\rho$, means $\mu_1, \mu_2$, standard deviations $\sigma_1, \sigma_2$.
Then
$$
X_1 |_{X_2 = x^{(2)}} \sim N(\mu, \sigma^2  ),
$$
where
$$
\mu = \mu_X + \rho \frac{\sigma_1}{\sigma_2}(x^{(2)} - \mu_2),
$$ $$
\sigma^2 = \sigma_1^2 (1 - \rho^2) .
$$
\end{newrule}

\begin{newrule}
\label{rule:KLdivergence.twoGaussians}
Let $p$ and $q$ be two Gaussian densities, $p\sim N(\mu_1, \sigma_1^2)$ and $q\sim N(\mu_2, \sigma_2^2)$.
Then the Kullback-Liebler (KL) Divergence between them is given by
$$
D(p || q) =
\log \left( \frac{\sigma_2}{\sigma_1} \right) +
\frac{1}{2} \left( \frac{\sigma_1^2 + (\mu_1 - \mu_2)^2}{\sigma_2^2}  \right) - \frac{1}{2} .
$$
\end{newrule}

\subsubsection{Auxiliary Computations}

In this section, we keep details for computations that use elementary methods of little technical interest.

\begin{computation}[The function $f(\gamma)$ in Eq~\ref{eqn:proof:linearIntxn.WB.specificInfo.specialFunction} is increasing.]
\label{computation.linearIntxn.WB.specificInfo.differenceFunctionIncreasing}
For any $\rho \in (-1,1)$, the function $f(\gamma)$ from Eq~\ref{eqn:proof:linearIntxn.WB.specificInfo.specialFunction}, reproduced below, is increasing in $\gamma$ on $[1, \infty)$:

\begin{align*}
\tag{\ref{eqn:proof:linearIntxn.WB.specificInfo.specialFunction}}
     f(\gamma) &=  \log \gamma
- \frac{\log(e) (1-\rho^2)(\gamma^2 - 1)}{2 (\gamma^2 + 2 \rho \gamma + 1)}
\end{align*}
\end{computation}

\begin{proof}
When we take the derivative of this function with respect to $\gamma$, we have
\begin{align*}
f'(\gamma)
&=\frac{4 (u(\gamma))^2-
4\gamma^2 (1-\rho^2) u(\gamma) + v(\gamma) }{4\gamma (u(\gamma))^2 },\\
\intertext{where}
u(\gamma) &= \gamma^2 + 2 \rho \gamma + 1,\\
v(\gamma) &= 4 \gamma(1-\rho^2)(\gamma^2 - 1)(\rho+\gamma).
\end{align*}
For $\gamma > 1$, note that $u(\gamma)$ and $v(\gamma)$ are positive.
Thus, positivity of $f'(\gamma)$ then follows from 
\begin{align*}
& 4 (u(\gamma))^2-
4\gamma^2 (1-\rho^2) u(\gamma) \\
&=u(\gamma)(4+8\rho\gamma+4\gamma^2-4\gamma^2+4\gamma^2\rho^2)\\
&= 4 u(\gamma) (1+\rho\gamma)^2 > 0.
\end{align*}

\end{proof}

\subsection{Figures for Network Simulations}

\begin{figure*}[b!]
    \includegraphics[scale=0.45]{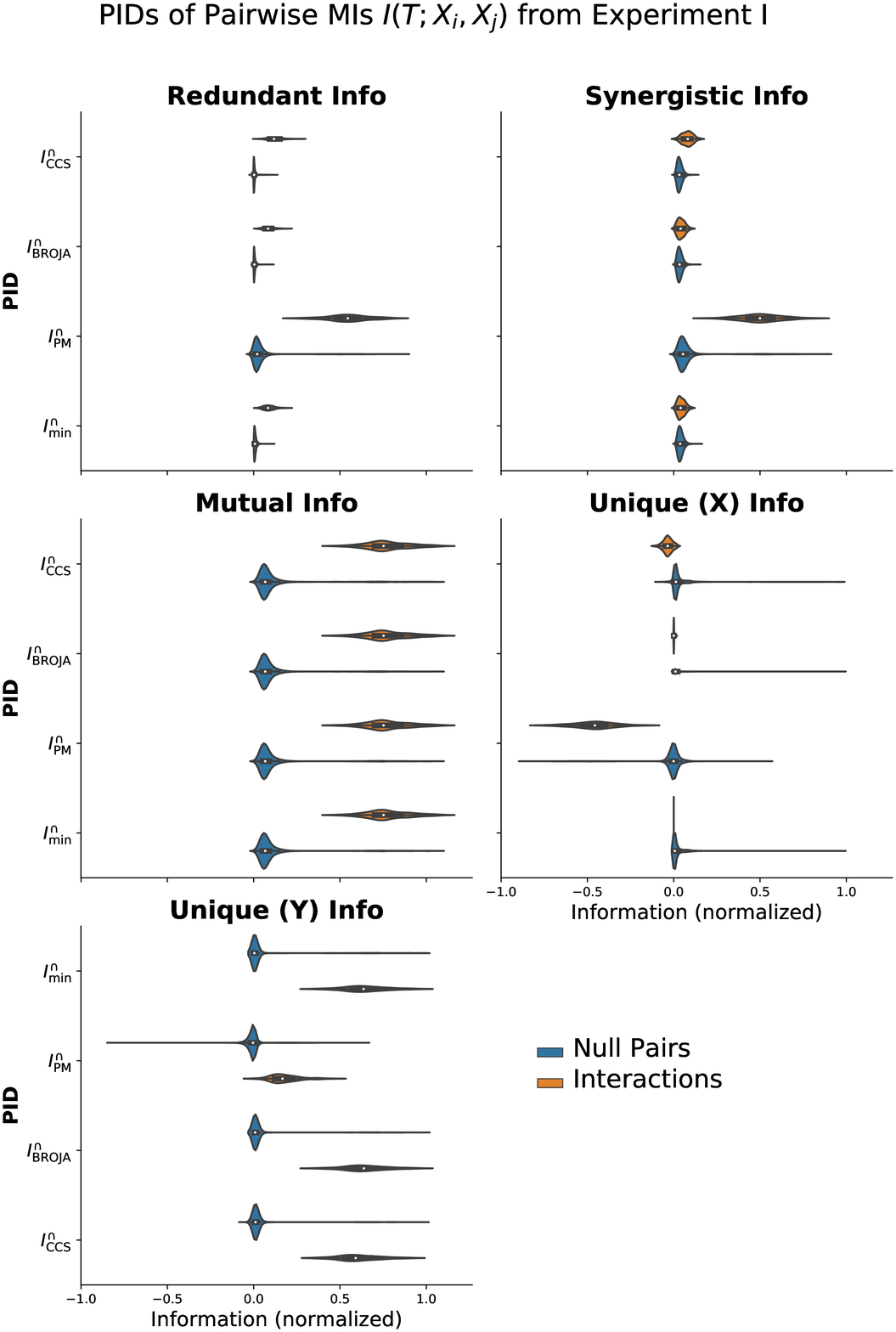}
    \caption{\textbf{Bivariate PID atoms for interactions in Experiment I} For the same bivariate PID data as in Fig.~\ref{fig:experiment1}, we present the full distributions of all PID atoms and mutual information, for both true interactions $(i,j) \in \mathcal{E}'$ and null pairs $(i,j) \notin \mathcal{E}'$. We can see that the $\Ipm$ PID assigns more synergy $\Spm$ to the interactions and negative unique information $\Upm{X}$ to the switch gene.}
    \label{fig:supp.expI}
\end{figure*}

\begin{figure*}
    \includegraphics[scale=0.4]{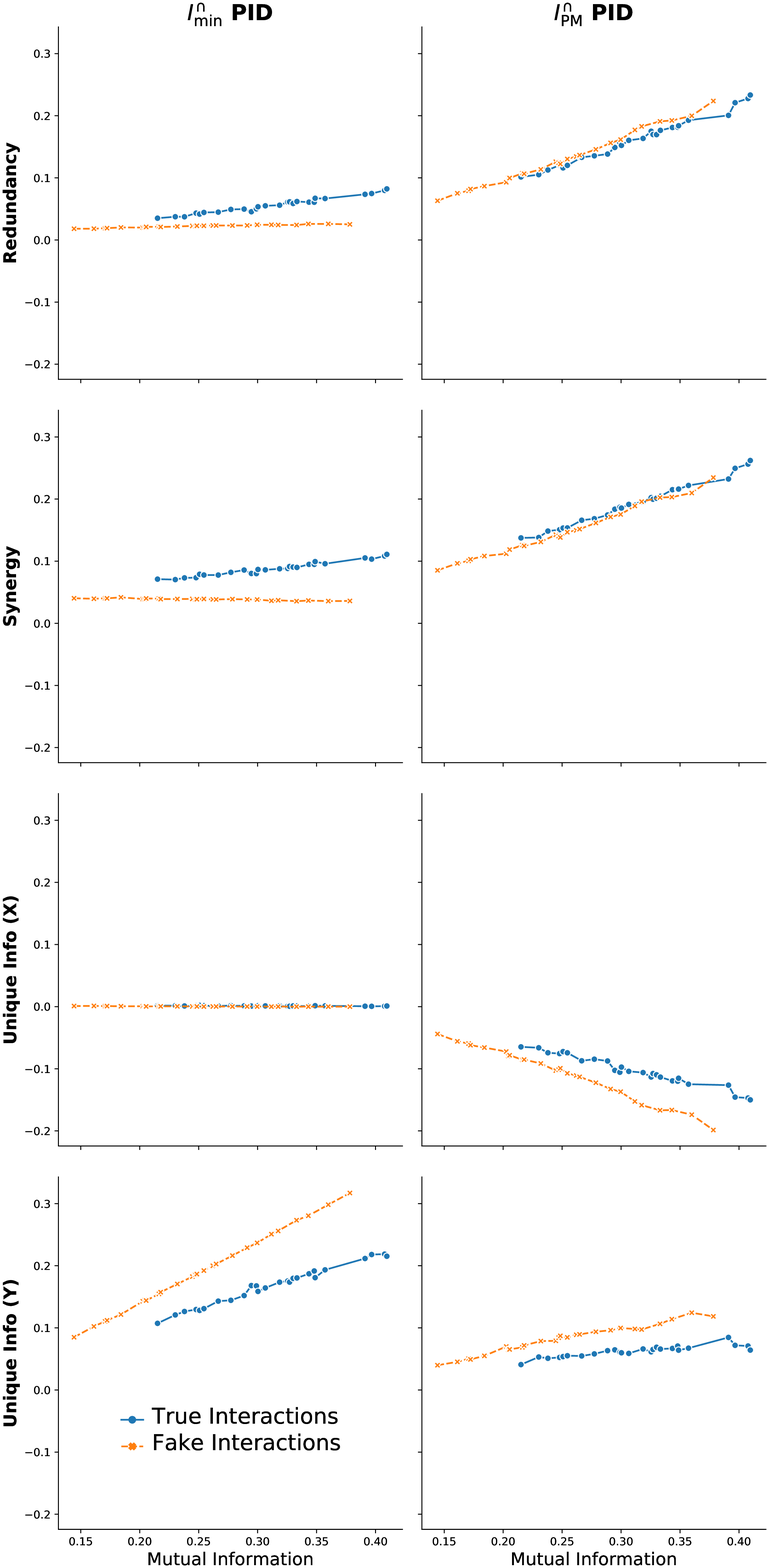}
    \caption{\textbf{Relationship between mutual information and bivariate $\Imin$ and $\Ipm$ atoms, for mixed interactions in Experiment~II.}
    In this figure, we display the relationship between mutual information and the (unranked) PID information atoms for Experiment II (Fig.~\ref{fig:experiment2}). Each connected scatter plot maps mutual information against a PID atom, for true and false interactions. Note that mutual information varies due to the parameter $\beta$ in Eq.~(\ref{eq:networkResponse.networkB}). We see that the $\Ipm$ PID tends to track closely with mutual information, with true and false interactions following a near-identical pattern for both synergy and redundancy. By contrast, the $\Imin$ PID is consistent in its assignment of the $\Rmin, \Smin,$ and $\Umin{X}$ atoms. Only the unique information assigned to the gene $Y_2$, which responsible for the univariate response signal $\beta Y_2$, increases with mutual information (corresponding to increasing $\beta$).
    }
    \label{fig:supp.expII}
\end{figure*}

\begin{figure*}
     \includegraphics[scale=0.4]{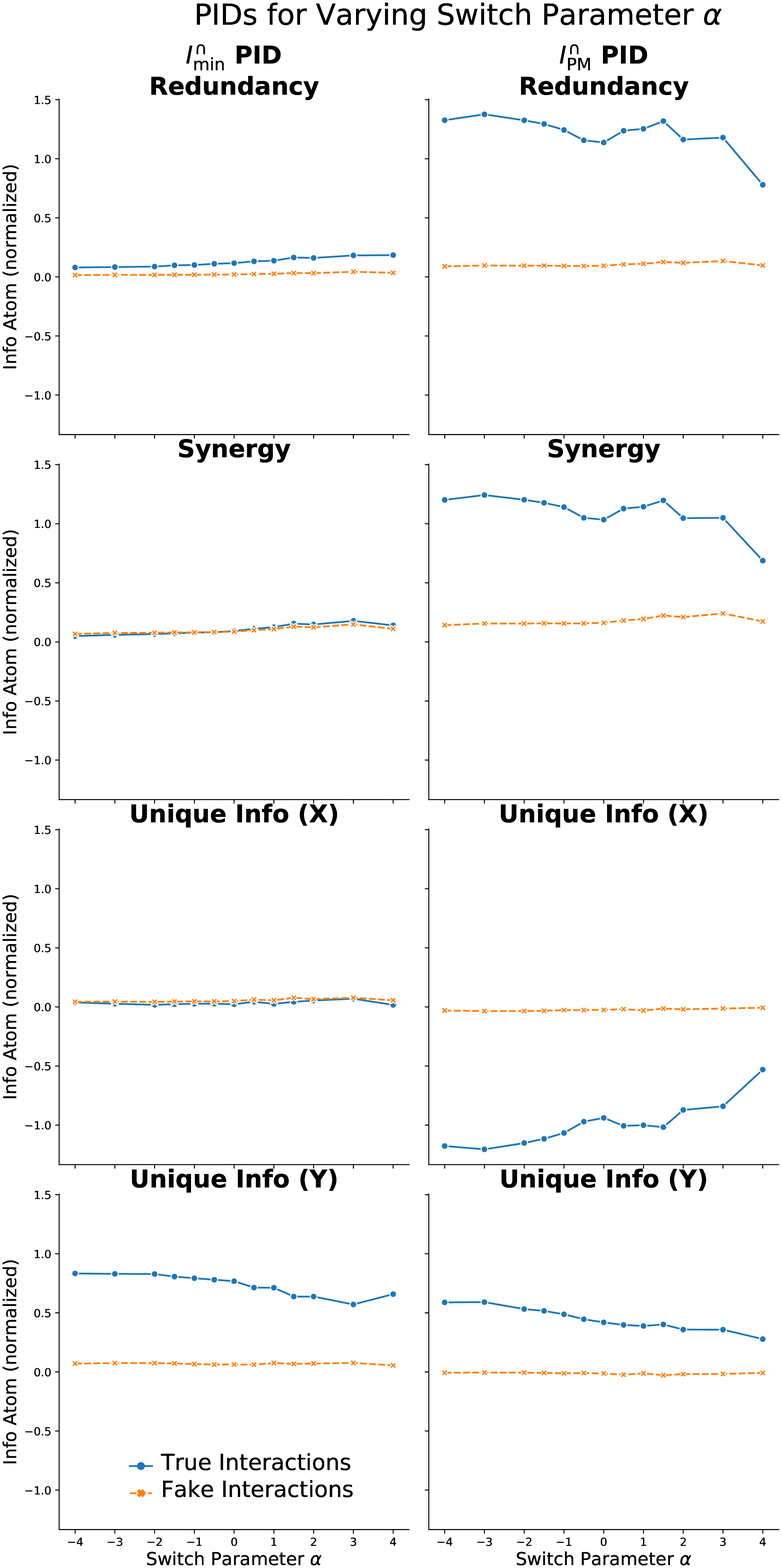}
     \caption{
     \textbf{Bivariate PID atoms (normalized) as a function of the switch parameter $\alpha$ from Experiment~III}
     We present the four PID atoms of the $\Imin$ and $\Ipm$ PIDs for Exp.~III, in which we vary the parameter $\alpha$ in the response (Eq.~\ref{eq:networkResponse.networkA}). Here, each atom has been normalized by the average MI of the interaction pairs, so as to discount for the effect of total MI varying with $\alpha$.}
     \label{fig:supp.expIII}
 \end{figure*}

\end{document}